\documentclass[reqno]{article}

\usepackage{amsthm}
\usepackage{amsmath}
\usepackage{amssymb}
\usepackage{graphicx}
\usepackage{color}
\usepackage{bbm}
\usepackage[top=2.5cm, bottom=2.5cm, left=2.5cm, right=2.5cm]{geometry}
\usepackage{enumerate}
\usepackage{faktor}
\usepackage{nicefrac}
\usepackage[nottoc]{tocbibind}
\usepackage{slashed}
\usepackage{authblk}

\usepackage[scr=euler]{mathalfa}

\newcommand{\N}{\mathbb{N}}
\newcommand{\R}{\mathbb{R}}

\newcommand{\CH}{\mathcal{CH}^+}
\newcommand{\Hp}{\mathcal{H}^+}
\newcommand{\rred}{r_\mathrm{red}}
\newcommand{\rblue}{r_\mathrm{blue}}
\newcommand{\T}{\mathbb{T}}
\newcommand{\M}{\overline{\mathcal M}}
\newcommand{\vol}{\mathrm{vol}}
\newcommand{\V}{{\mathrm{V}}}
\newcommand{\volsp}{\mathrm{vol}_{f^+}}
\newcommand{\volsn}{\mathrm{vol}_{f^-}}
\newcommand{\volg}{\mathrm{vol}_{f_{\gamma_\sigma}}}
\newcommand{\volr}{\mathrm{vol}_{-r}}
\newcommand{\Co}{\mathscr{C}}
\newcommand{\Si}{\mathscr{S}}

\newcommand{\f}{\frac}
\newcommand{\rd}{\partial}

\begin{document}

\numberwithin{equation}{section}
\newtheorem{theorem}[equation]{Theorem}
\newtheorem{remark}[equation]{Remark}
\newtheorem{assumption}[equation]{Assumption}
\newtheorem{claim}[equation]{Claim}
\newtheorem{lemma}[equation]{Lemma}
\newtheorem{definition}[equation]{Definition}
\newtheorem{corollary}[equation]{Corollary}
\newtheorem{proposition}[equation]{Proposition}
\newtheorem*{theorem*}{Theorem}
\newtheorem{conjecture}[equation]{Conjecture}

\title{Instability results for the wave equation in the interior of Kerr black holes}

\author[1]{Jonathan Luk\thanks{jluk@dpmms.cam.ac.uk}}
\author[2]{Jan Sbierski\thanks{jjs48@cam.ac.uk}}
\affil[1,2]{Department for Pure Mathematics and Mathematical Statistics, University of Cambridge,
Wilberforce Road,
Cambridge,
CB3 0WA,
United Kingdom}

\date{\today}

\maketitle

\begin{abstract}
We prove that a large class of smooth solutions $\psi$ to the linear wave equation $\Box_g\psi=0$ on subextremal rotating Kerr spacetimes which are regular and decaying along the event horizon become singular at the Cauchy horizon. More precisely, we show that assuming appropriate upper and lower bounds on the energy along the event horizon, the solution has infinite (non-degenerate) energy on any spacelike hypersurfaces intersecting the Cauchy horizon transversally. Extrapolating from known results in the Reissner--Nordstr\"om case, the assumed upper and lower bounds required for our theorem are conjectured to hold for solutions arising from generic smooth and compactly supported initial data on a Cauchy hypersurface. This result is motivated by the strong cosmic censorship conjecture in general relativity.
\end{abstract}

\tableofcontents

\section{Introduction}

In this paper, we study the linear instability of the Kerr Cauchy horizon under scalar perturbations. More precisely, we consider the linear wave equation\footnote{Here, $\Box_g$ denotes the standard Laplace--Beltrami operator for the Kerr metric $g$.} on the black hole interior of subextremal Kerr spacetime with non-vanishing angular momentum (i.e., for Kerr parameters $0< |a|< M$ - see \eqref{Kerr.metric.BL} for the formula for the metric)
\begin{equation}\label{linear.wave}
\Box_g\psi = 0
\end{equation}
with (characteristic) initial data posed on the event horizon. We prove that there exists a large class of smooth, polynomially decaying\footnote{with respect to the function $v_+$ (see Section \ref{sec.geometry}).} data such that the solutions have infinite non-degenerate energy on any spacelike hypersurface intersecting the Cauchy horizon transversally. In particular, such solutions do not belong to $W^{1,2}_{loc}$ in any neighbourhood of any point on the Cauchy horizon. In addition to merely constructing singular solutions, our main theorem identifies a sufficient condition only in terms of upper and lower bounds of appropriate energy along the event horizon that guarantees the solution blows up at the Cauchy horizon. We state our main result roughly as follows and refer the readers to Theorem \ref{MainThm} for a precise statement.  
\begin{theorem}[Rough version of main theorem]\label{main.theorem.intro}
Let $\psi$ be a smooth solution to \eqref{linear.wave} in the interior of a subextremal Kerr spacetime with non-vanishing angular momentum. If the energy of $\psi$ along the event horizon obeys some polynomial upper and lower bounds (see i)-iii) {of} Theorem \ref{MainThm} for precise bounds), then the non-degenerate energy\footnote{See Remark \ref{rmk.energy} for further discussions on the non-degenerate energy.} on any spacelike hypersurface intersecting the Cauchy horizon transversally is infinite.
\end{theorem}
Our result is motivated by the celebrated strong cosmic censorship conjecture in general relativity. We will not discuss this conjecture in detail, but refer the readers to \cite{Daf14, DafLuk, LukOh15} for further discussions. For the purpose of this paper, it suffices to say that the maximal globally hyperbolic development of the Kerr solution for $0< |a|< M$ has a smooth Cauchy horizon and the strong cosmic censorship conjecture suggests the following instability conjecture of the Kerr Cauchy horizon for small perturbations of Kerr:
\begin{conjecture}\label{instability.conj}
Generic small perturbations of Kerr initial data for the Einstein vacuum equations
\begin{equation}\label{Einstein.vacuum}
Ric(g)_{\mu\nu}=0
\end{equation}
lead to maximal globally hyperbolic developments such that the Cauchy horizons are so-called weak null singularities. In particular, for \underline{any} continuous extensions of the spacetime metric, the spacetime Christoffel symbols are not square-integrable.
\end{conjecture}
Early formulations of this instability conjecture often suggest an even stronger instability - namely, that a spacelike singularity may form ``prior to the Cauchy horizon'' and that the spacetime does not contain a Cauchy horizon at all. However, recent work \cite{DafLuk} shows that the Kerr Cauchy horizon is in fact $C^0$-stable if one assumes that the exterior region of Kerr is stable is a reasonably strong sense. On the other hand, the estimates proven in \cite{DafLuk} are consistent with the spacetime not having square integrable Christoffel symbols. One can therefore still hope that the weaker formulation of the instability conjecture based on the non-square-integrability of Christoffel symbols may hold.\footnote{This formulation is due to Christodoulou \cite{Chris09}. In particular, if this formulation of the conjecture holds, then the maximal globally hyperbolic developments of generic initial data in a small neighbourhood of Kerr data admit no extensions as weak solutions to the Einstein vacuum equations.} 

Theorem \ref{main.theorem.intro} can be viewed as a first step towards establishing Conjecture \ref{instability.conj}. Instead of considering the full nonlinear Einstein vacuum equations, we only study a much simpler model equation, namely the linear scalar wave equation \eqref{linear.wave}. This can be regarded as a simplified linearization of the Einstein equations in which we ignore the tensorial structure of the system and drop all the lower order terms\footnote{Notice in particular that the linear equation that we study is \underline{not} a true linearization of the Einstein vacuum equations.}. We prove in Theorem \ref{main.theorem.intro} that at least in this much simpler setting, there is indeed an instability mechanism. If one moreover naively compares\footnote{Recall that in a generalized wave coordinate system, \eqref{Einstein.vacuum} becomes a wave equation for the metric $g$. We should remark however that this comparison is best taken only at a heuristic level, as a resolution of Conjecture \ref{instability.conj} will likely not be based on generalized wave coordinates (cf. \cite{DafLuk}). } the metric in \eqref{Einstein.vacuum} with the scalar field in \eqref{linear.wave}, then the infinitude of the non-degenerate energy of $\psi$ corresponds to the blow-up of the $L^{2}_{loc}$ norm of the Christoffel symbols as in Conjecture \ref{instability.conj}.

There is a long tradition in both mathematics and physics in studying the linear scalar wave equation on black hole backgrounds. The existence of solutions to the linear wave equation on Kerr which are regular at the event horizon but singular at the Cauchy horizon have been previously constructed in \cite{McNam78, Sbie13b} (see also \cite{DafShla15}). Therefore, the main novelty of Theorem \ref{main.theorem.intro} is that it gives a sufficient condition for the solution to be singular at the Cauchy horizon only in terms of $L^2$ upper and lower bounds of the solution along the event horizon. Moreover, motivated by the known results in the case of the subextremal Reissner--Nordstr\"om spacetime\footnote{See Section \ref{sec.outline} for the metric of Reissner--Nordstr\"om spacetime and discussions on its geometry.} \footnote{One can view the subextremal Reissner--Nordstr\"om spacetime (with non-vanishing charge) as a ``poor-man's version'' of the subextremal Kerr spacetime (with non-vanishing angular momentum) in the sense that they have similar global geometries, including having smooth Cauchy horizons, while the geometry in the Reissner--Nordstr\"om case is simpler. Indeed, as we will discuss in Section \ref{sec.global}, more is known about solutions to the linear wave equation in Reissner--Nordstr\"om. Moreover, in that case, there exists solutions arising from smooth and compactly supported Cauchy data such that analogues of the assumptions of Theorem \ref{main.theorem.intro} hold (see Section \ref{sec.global}).} with non-vanishing charge, one may {\bf conjecture} that the bounds that are needed in the assumptions of Theorem \ref{main.theorem.intro} hold for some solutions to the linear wave equation arising from smooth and compactly supported initial data on a complete 2-ended Cauchy hypersurface. This conjecture, if true, means that generic smooth and compactly supported initial data on a Cauchy hypersurface lead to solutions that are singular at the Cauchy horizon. We will return to this point in the discussions in Section \ref{sec.global}.

Our proof of Theorem \ref{main.theorem.intro} is based on energy identities and energy estimates in the interior of the Kerr black hole, which are established by appropriate choices of various vector fields. Most importantly, we rely on the conservation law\footnote{Recall Noether's theorem which allows us to obtain a conservation law for the solutions to a wave equation from a symmetry of the underlying spacetime.} associated to a Killing vector field $T_{\mathcal C\mathcal H^+}$ (see \eqref{TCH.def} for the definition) which extends the null generators of the Cauchy horizon. This conservation law lets us propagate an $L^2$ lower bound from the event horizon up to a spacelike hypersurface $\gamma_\sigma$ (see \eqref{gamma.def} for the definition) that approaches the Cauchy horizon ``at timelike infinity''. For an appropriate choice of the hypersurface $\gamma_\sigma$, we can then propagate this lower bound to a hypersurface transversally intersecting the Cauchy horizon using the energy identity associated to the (non-Killing!) vector field $\f{\Delta}{\rho^2}\rd_r+\f{r^2+a^2}{\rho^2}\rd_t+\f{a}{\rho^2}\rd_\varphi$. In order to\footnote{In fact, stability estimates are also needed to bound the ``boundary terms'' in the conservation law associated to $T_{\mathcal C\mathcal H^+}$ (see discussions in Section \ref{sec.outline}).} control the error terms arising from this energy identity, we need to prove sufficiently strong \emph{stability} estimates, which in turn are also obtained via energy estimates. We refer the readers to further discussions on the ideas of the proof in Section \ref{sec.outline} where we sketch the argument in the slightly simpler case of the Reissner--Nordstr\"om spacetime.

In the remainder of the introduction, we will discuss the global Cauchy problem for \eqref{linear.wave} in Section \ref{sec.global}. In particular, we will point out the relevance of Theorem \ref{main.theorem.intro} to the global Cauchy problem. Then, in Section \ref{sec.outline}, we will outline some ideas of the proof in the setting of the Reissner--Nordstr\"om spacetime. Finally, we will end the introduction with an outline of the remainder of the paper.

\subsection{The global Cauchy problem}\label{sec.global}

With an eye towards the instability of Kerr Cauchy horizon conjecture (Conjecture \ref{instability.conj}), one would like to go beyond the study of solutions to the linear wave equation within the interior of the black hole region as in Theorem \ref{main.theorem.intro}, but instead consider the problem of global evolution, where the initial data are prescribed on a $2$-ended complete asymptotically flat Cauchy hypersurface $\Sigma_0$ (see\footnote{Here, the interior of black hole region depicted in Figure \ref{FigInt} on page \pageref{FigInt} should be thought of as the top ``diamond'' region in Figure \ref{FigGlobal}.} Figure \ref{FigGlobal}).

\begin{figure}[h]
  \centering
  \def\svgwidth{6cm}
    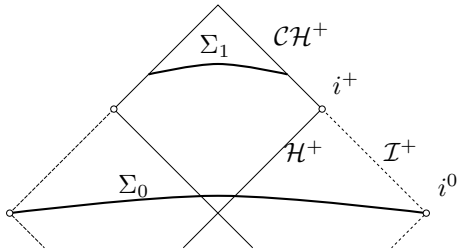
      \caption{The global Cauchy problem} \label{FigGlobal}
\end{figure}

Due to the blue-shift effect associated with the Cauchy horizon, it has been expected since the pioneering work of Simpson--Penrose \cite{SiPen73} that there is a global instability for the Cauchy problem such that generic regular data give rise to solutions which are singular at the Cauchy horizon. This problem has been widely studied in the physics literature, see for instance \cite{ChandHart82, GSNS79, MZS, McNam78.2}.

The first rigorous result on the instability of solutions to the linear wave equation dates back to McNamara \cite{McNam78}, who gave a conditional\footnote{Notice that in the case of Reissner-Nordstr\"om, the condition also was shown to hold in \cite{McNam78}.} proof of the existence of solutions arising from regular polynomially decay{ing} data but are singular at the Cauchy horizon, subject to verifying a condition regarding the non-triviality of the ``scattering map'' of the solutions from past null infinity to the the Cauchy horizon. This condition has recently been proven by Dafermos--Shlapentokh-Rothman \cite{DafShla15}, who also gave an alternative proof of this result as part of their general treatment of blue-shift instabilities on black hole spacetimes. We give a rough version of the theorem here and refer the readers to \cite{DafShla15} for a precise statement:
\begin{theorem}[McNamara \cite{McNam78}, Dafermos--Shlapentokh-Rothman \cite{DafShla15}]\label{McNam.thm}
On any subextremal Kerr spacetime with non-vanishing angular momentum, there exist solutions to the linear wave equation which arise from polynomially decaying data (with an arbitrarily fast polynomial rate) on past null infinity but are singular at the Cauchy horizon.
\end{theorem}

Another construction of solutions which are singular at the Cauchy horizon but have finite initial energy follows from the results of Sbierski. More precisely, the following was proven in \cite{Sbie13b} by geometric optics considerations:
\begin{theorem}[Sbierski \cite{Sbie13b}]\label{Sbie.thm}
Consider a subextremal Kerr spacetime with non-vanishing angular momentum. Let $\Sigma_0$ be a complete $2$-ended asymptotically flat Cauchy hypersurface and $\Sigma_1$ be a spacelike hypersurface in the interior of the black hole intersecting the Cauchy horizon transversally (see Figure \ref{FigGlobal}). There exists a sequence of solutions $\{\psi_i\}_{i=1}^{\infty}$ to the linear wave equation $\Box_g\psi_i=0$ such that the initial energies on $\Sigma_0$ satisfy $E(\psi_i,\Sigma_0)=1$ while the non-degenerate energies on $\Sigma_1$ obey $E(\psi_i,\Sigma_1)\to +\infty$.
\end{theorem}
Given the above result, an application of the closed graph theorem\footnote{{More precisely, we use the following consequence of the closed graph theorem: Let $X$, $Y$, $Z$ be Banach spaces and $T:X\to Y$, $Q:Y\to Z$ be linear maps. If $Q$ is bounded and injective and $Q\circ T$ is bounded, then $T$ is bounded. In our setting, let $X$ be the energy space on $\Sigma_0$, $Y$ and $Z$ be the non-degenerate and \emph{degenerate} energy spaces on $\Sigma_1$ respectively. \emph{If for the sake of argument that all solutions initially in the energy class on $\Sigma_0$ belong to the non-degnerate energy class on $\Sigma_1$}, then we can define $T$ to be the operator ``solving the wave equation'' and $Q$ to be the inclusion map and derive a contradiction with Theorem \ref{Sbie.thm} using the above functional analytic statement. The key remaining analytic ingredient is the boundedness of $Q\circ T$, i.e., the boundedness of the \emph{degenerate} energy, which follows from \cite{DafRodShla14} together with considerations in Section \ref{sec.stability}.}} implies the existence of solutions which initially belong to the energy class but then fail to be in the non-degenerate energy class near the Cauchy horizon.

In the constructions in Theorems \ref{McNam.thm} and \ref{Sbie.thm} above, the solutions do not have compact support on the Cauchy hypersurface $\Sigma_0$.  As pointed out in \cite{Sbie13b}, Theorem \ref{Sbie.thm} also holds in the extremal case, i.e., when $0<  |a|=M$. In that case, however, the instability may in fact be milder - it is known that for {\bf axisymmetric} data on extremal Kerr spacetimes which decay sufficiently fast on a Cauchy hypersurface\footnote{By saying this, we have chosen a globally hyperbolic subset of the extremal Kerr spacetime for which there is an (incomplete!) asymptotically flat Cauchy hypersurface which extends into the black hole region. For precise statements, see \cite{Sbie13b, Gajic15b}.}, the corresponding solutions to the linear wave equation have finite non-degenerate energy in the interior of the black hole \cite{Gajic15b} (see also \cite{Gajic15a, MuReTa13}). In other words, the analogue of Theorem \ref{Sbie.thm} in the extremal case only holds because one considers general Cauchy data in the energy class. On the other hand, in the case of subextremal Kerr spacetime, it is indeed expected that the singularity at the Cauchy horizon is not an artefact of the slow decay of the initial data at spacelike infinity. It is therefore desirable {to} obtain an instability result for smooth and compactly supported Cauchy data. We formulate this as a conjecture:
\begin{conjecture}\label{global.conj}
Generic smooth and compactly supported Cauchy data on $\Sigma_0$ to the linear wave equation \eqref{linear.wave} in subextremal rotating Kerr spacetimes give rise to solutions such that the non-degenerate energy on any spacelike hypersurface intersecting the Cauchy horizon transversally is infinite.
\end{conjecture}

While the above conjecture remains an open problem, an analogue in the Reissner--Nordstr\"om spacetime has been recently established:
\begin{theorem}[Luk--Oh \cite{LukOh15}]\label{LO.thm}
Given any {\bf Reissner--Nordstr\"om spacetime} with $0< |e|< M$, generic smooth and compactly supported initial data on a complete $2$-ended asymptotically flat Cauchy hypersurface $\Sigma_0$ (Figure\footnote{Figure \ref{FigGlobal} was of course meant to depict the Kerr spacetime. However, since Kerr spacetime (for $0< |a|< M$) and Reissner--Nordstr\"om spacetime (for $0< |e|< M$) can be described by the same Penrose diagram, we also use it to depict the Reissner--Nordstr\"om spacetime for the purpose of the statement of this theorem.} \ref{FigGlobal}) give rise to solutions to the linear wave equation such that the non-degenerate energy on any spacelike hypersurface intersecting the Cauchy horizon transversally is infinite.
\end{theorem}

This theorem was proved in \cite{LukOh15} via first showing that a lower bound\footnote{Although it was not explicitly used in \cite{LukOh15}, upper bounds analogous to those in Assumptions i) and iii) in Theorem \ref{MainThm} are also known in the Reissner-Nordstr\"om setting \cite{DafRod05, MeTaTo12, Tat13}.} on the energy along the event horizon holds for a generic class of solutions arising from smooth and compactly supported data on $\Sigma_0$. This lower bound is in a similar form as that required for our main theorem (see assumption ii) of Theorem \ref{MainThm}). Based on the known results in the Reissner--Nordstr\"om case, one can therefore hope that in the Kerr case, there exist smooth and compactly supported data on $\Sigma_0$ such that the solutions obey the assumptions of Theorem \ref{main.theorem.intro} along the event horizon. In other words, by Theorem \ref{main.theorem.intro}, Conjecture \ref{global.conj} can be reduced\footnote{Strictly speaking, the combination of Theorem \ref{main.theorem.intro} and Conjecture \ref{Kerr.existence.conj} only proves the \emph{existence} of singular solutions arising from smooth and compactly supported initial data. Nevertheless, by virtue of the linearity of the wave equation, this would immediately implies the \emph{genericity} of such singular solutions.} to the following conjecture:
\begin{conjecture}\label{Kerr.existence.conj}
There exist smooth and compactly supported initial data on a complete $2$-ended asymptotically flat Cauchy hypersurface $\Sigma_0$ (Figure \ref{FigGlobal}) in {\bf Kerr spacetime} (with $0< |a|< M$) such that the solutions to the linear wave equation \eqref{linear.wave} obey the assumptions on the energy along the event horizon in Theorem \ref{main.theorem.intro}.
\end{conjecture}
Let us finally remark that the upper bounds i) and iii) in the assumptions of Theorem \ref{MainThm} are known to hold for some $q> 0$ \cite{DafRod09b, DafRodShla14, MeTaTo12, Tat13}. Thus the main challenge is to obtain the lower bound ii). Moreover, one needs to show near optimal upper and lower bounds such that i) and ii) in the assumptions of Theorem \ref{MainThm} hold with the same $q> 0$.

\subsection{Outline of the proof by example of the subextremal Reissner--Nordstr\"om black hole interior}\label{sec.outline}

In this section, we discuss the main ideas of the proof of our theorem. In order to describe the key points of the proof without getting into the technical complications arising from the geometry of Kerr, we will restrict our attention to Reissner--Nordstr\"om in this section. We first briefly recall the geometry of subextremal Reissner--Nordstr\"om black hole interior in the case of non-vanishing charge. For $0< |e|< M$, the interior of the Reissner--Nordstr\"om black hole is the manifold $\mathcal M_{RN}=\R \times (r_-,r_+) \times \mathbb{S}^2$ with the following metric:
$$g_{RN}=-(1-\f{2M}{r}+\f{e^2}{r^2})\,dt^2+(1-\f{2M}{r}+\f{e^2}{r^2})^{-1} \,dr^2+r^2 (d\theta^2+\sin^2\theta d\varphi^2),$$
where $r_{\pm}=M\pm\sqrt{M^2-e^2}$ are the two roots of the polynomial $\Delta_{RN}:=r^2-2Mr+e^2$. Define $r^*=r+(M+\frac{2M^2-e^2}{2\sqrt{M^2-e^2}})\log |r-r_+| +(M-\frac{2M^2-e^2}{2\sqrt{M^2-e^2}})\log (r-r_-)$, $v_+=t+r^*$, $v_-=r^*-t$. 
 Notice that $v_+$ and $v_-$ are null variables\footnote{Note that this is in contrast to the definitions of $v_+$ and $v_-$ that we introduce in the Kerr case, which will not be null variables (see Section \ref{sec.geometry})}. The metric takes the form
 \begin{equation*}
 g_{RN} = \frac{\Delta_{RN}}{2 r^2}\,(dv_+ \otimes dv_- + dv_- \otimes dv_+) + r^2 \,(d\theta^2 + \sin^2 \theta \, d\varphi^2)\;.
 \end{equation*}

We then attach the boundaries $\mathcal H^+$ and $\mathcal CH^+$ to $\mathcal M_{RN}$, where $\mathcal H^+$ is the set $\mathbb R\times \{r=r_+\}\times \mathbb S^2$ in the $(v_+,r,\theta,\varphi)$ coordinate system and $\mathcal C\mathcal H^+$ is the set $\mathbb R\times \{r=r_-\}\times \mathbb S^2$ in the $(v_-,r,\theta,\varphi)$ coordinate system (see Figure \ref{FigInt}). 

Below, we will consider various vector fields using the following conventions: $\rd_t$ is to be understood as the coordinate vector field in the $(t,r,\theta,\varphi)$ coordinate system, while $\rd_{v_+}$ and $\rd_{v_-}$ are to be understood as the coordinate vector fields in the $(v_+,v_-,\theta,\varphi)$ coordinate system. Define also the vector fields $(\Omega^1, \Omega^2, \Omega^3):=(-\f{\cos\theta\cos\varphi}{\sin\theta}\rd_{\varphi}-\sin\varphi\rd_{\theta},-\f{\cos\theta\sin\varphi}{\sin\theta}\rd_{\varphi}+\cos\varphi\rd_{\theta},\rd_{\varphi})$ and the notation that $|\slashed\nabla\psi|^2=\f{1}{r^2}\sum_{i=1}^3|\Omega^i\psi|^2$.

Using the above notations, a precise version of the analogue of Theorem \ref{main.theorem.intro} in the setting of Reissner--Nordstr\"om black hole interior is given as follows\footnote{Notice that in Theorem \ref{main.theorem.RN}, a lower bound which implies the infinitude of the non-degenerate energy is proven on a null hypersurface. This is in contrast to a lower bound on a spacelike hypersurface in Theorem \ref{MainThm}. This is simply carried out for expositional convenience and it is easy to pass from lower bounds on null hypersurfaces to lower bounds on spacelike hypersurfaces.}:
\begin{theorem}\label{main.theorem.RN}
Let $\psi:\mathcal M_{RN}\cup \mathcal H^+ \to \R$ be a smooth solution to the wave equation $\Box_{g_{RN}}\psi=0$. Assume that there exists some $q> 0$ and $\delta\in[0,1)$ such that
\begin{enumerate}[i)]
\item the following upper bound holds\footnote{It is understood here, and below, that the implicit constants are independent of $\V$.} on $\mathcal H^+$ for all $\V\geq 1$,
$$\int\limits_{\mathcal H^+\cap\{v_+\geq \V\}}  \left((\rd_{v_+}\psi)^2+|\slashed\nabla\psi|^2\right)\,r^2\sin^2\theta\, dv_+\,d\theta\,d\varphi\lesssim \V^{-q};$$
\item the following lower bound holds on $\mathcal H^+$ for all $\V\geq 1$,
$$\int\limits_{\mathcal H^+\cap\{v_+\geq \V\}} (\rd_{v_+}\psi)^2 r^2\sin^2\theta\, dv_+\,d\theta\,d\varphi \gtrsim \V^{-(q+\delta)};$$
\item the following upper bound holds for the higher derivative of $\psi$ on $\mathcal H^+$,
$$\sum_{i=1}^3\int\limits_{\mathcal H^+\cap\{v_+\geq 1\}}  \left((\rd_{v_+}\Omega^i\psi)^2+|\slashed\nabla\Omega^i\psi|^2\right)\,r^2\sin^2\theta\, dv_+\,d\theta\,d\varphi\lesssim 1.$$
\end{enumerate}
Then for every $u_0\in \mathbb R$, there exists a sequence $v_k\in \mathbb R$ with $v_k\to \infty$ as $k\to \infty$ such that\footnote{{Using the coordinates $(v_-, r, \theta, \varphi)$, which are regular at $\CH$, it is seen that $\frac{1}{(-\Delta_{RN})}\partial_{v_+}$ is regular at $\CH$. Moreover, for constant $v_-$, $(-\Delta_{RN})$ decays exponentially in $v_+$.
Hence, the boundedness of the non-degenerate energy on $\{v_-=u_0\}$ would imply that the left hand side of \eqref{PolBound} decays at least exponentially in $v_k$. See the discussions in the Kerr case in \eqref{NotL2} and Remarks \ref{Rmk.NotL2} and \ref{rmk.energy} for more details. Hence, \eqref{PolBound} implies in particular the blow-up of the non-degenerate energy on $\{v_- = u_0\}$.}}
\begin{equation}
\label{PolBound}
\int\limits_{\{v_-=u_0\}\cap\{v_+\geq v_k\}} (\rd_{v_+}\psi)^2 r^2\sin^2\theta\, dv_+\,d\theta\,d\varphi \gtrsim (v_k)^{-(q+\delta)}.
\end{equation}
\end{theorem}
A slight variant of this theorem was proven in \cite{LukOh15} if $\psi$ is in addition assumed to be {\bf spherically symmetric}. In other words, the methods in this paper give an extension to the corresponding result in \cite{LukOh15} to allow for general $\psi$. In fact, even when restricted to spherical symmetry, the present paper provides an alternative approach to that in \cite{LukOh15}.

Our strategy is to carry out the proof in the following steps:
\begin{enumerate}
\item Stability estimates,
\item Instability estimates up to the hypersurface $\gamma_\sigma$ via the $\rd_t$-conservation law,
\item Instability estimates to the future of $\gamma_\sigma$.
\end{enumerate}

In each of these steps the key is the following energy identity, which holds for any solutions to the linear wave equation in any compact region $D \subseteq \M$ with piecewise smooth boundary $\partial D$, which is oriented with respect to the outward pointing normal
\begin{equation}\label{EE.intro}
\int\limits_{\partial D} \iota_{\T[\psi](Z, \cdot)^\sharp} \,\vol = \int\limits_{D} d\big(  \iota_{\T[\psi](Z, \cdot)^\sharp} \,\vol\big) = \int\limits_D \T[\psi]_{\mu \nu} \nabla^\mu Z^\nu  \, \vol \;,
\end{equation}
where $\T[\psi]$ is the stress-energy-momentum tensor given by
\begin{equation*}
\T[\psi]_{\mu \nu} := \partial_\mu \psi \partial_\nu \psi - \frac{1}{2} g_{\mu \nu} g^{-1}(d\psi, d\psi) .
\end{equation*}
The derivation of the energy identity relies on the fact $\T$ is divergence-free by virtue of the linear wave equation. We refer the readers to Section \ref{sec.energy.estimates} for further discussions.

\subsubsection{Stability estimates}

We now explain each of the steps above. The first step, i.e., the stability estimates, is already carried out in \cite{Fra14}. Since we need a slightly different version, we state it here in Proposition \ref{RN.stab}, with a brief sketch of the proof. The complete proof will be carried out in the Kerr case in Section \ref{sec.stability}, and in the sketch below, we will point out where the analogue of each of the steps in the Kerr case will be carried out in the paper.

Before we state the proposition on stability estimates, we first need to define a hypersurface in the interior of Reissner-Nordstr\"om, which plays a crucial role in the analysis. Define a function
$$f_{\gamma_\sigma} (v_+, v_-) := v_+ + v_- - \sigma \log(v_+) $$
for $v_+$ sufficiently large and $\sigma> 0$ and define a hypersurface $\gamma_{\sigma}$ in the interior of the Reissner--Nordstr\"om black hole by
\begin{equation}\label{gamma.sigma.RN.def}
\gamma_\sigma := f_{\gamma_\sigma}^{-1}(1).
\end{equation}
An analogue of this hypersurface was first introduced by Dafermos \cite{Daf03, Daf05a} in the setting of the Einstein--Maxwell--(real)--scalar--field system in spherical symmetry (see also \cite{Fra14}). This hypersurface has the important property that its future, restricted to the past of $\{v_-=u_0\}$ (for arbitrary $u_0$), has finite spacetime volume. This fact will be the underlying geometric reason that the error terms in Proposition \ref{prop.instability.2.RN} are under control.

The following are the main stability estimates:

\begin{proposition}\label{RN.stab}
Let $\alpha\in [0,1)$, $r_0\in (r_-,r_+)$ and $u_1\in \mathbb R$. Denote by $\vol$ the metric volume form
$$\vol=\frac{1}{2}(-\Delta_{RN}) \sin\theta\,dv_-\,dv_+\,d\theta\,d\varphi.$$
Then there exists $C> 0$ such that for all $V\geq 1$, the following stability estimates hold for $\psi$ satisfying i) and iii) in the assumptions of Theorem \ref{main.theorem.RN}:
\begin{equation}\label{ILED.RN}
\begin{split}
&\underbrace{\int\limits_{\{v_+ \geq \V\}  \cap \{ r\geq r_0\} } \Big( (\rd_{v_+} \psi)^2 + \f{1}{(-\Delta_{RN})^2}(\rd_{v_-} \psi)^2 + |\slashed\nabla \psi|^2 \Big) \, \vol}_{=:IE_1[\psi;V]} \\
&+ \underbrace{\int\limits_{\substack{\{v_+ \geq \V\} \cap\{f_{\gamma_\sigma}\leq 1\}\\ \cap \{ r\leq r_0\} }} \Big( \frac{1}{(-\Delta_{RN})^\alpha} \big[(\rd_{v_+} \psi)^2 + (\rd_{v_-} \psi)^2\big] + |\slashed\nabla \psi|^2 \Big)\,  \vol}_{=:IE_2[\psi;V]}\leq C \V^{-q}\;,
\end{split}
\end{equation}
and
\begin{equation}
\label{ImprovedILED.RN}
\int\limits_{\{ r\leq r_0\} \cap \{v_- \leq u_1\}} \frac{1}{(-\Delta_{RN})^\alpha} \Big( (\rd_{v_+}\psi)^2 + (\rd_{v_-}\psi)^2 +  |\slashed\nabla\psi|^2 \Big)\, \vol \leq C \;.
\end{equation}
\end{proposition}
\begin{proof}[Sketch of proof]

\underline{Step One (Section \ref{sec.ILED.1st})} The first step is to establish the following integrated energy estimates to the past of the hypersurface $\{r=r_0\}$:
\begin{equation}\label{RN.step1}
\begin{split}
& IE_{1}[\psi;\V]+\int\limits_{\{v_+\geq V\}\cap\{r=r_0\}} \left((\rd_{v_+}\psi)^2+(\rd_{v_-}\psi)^2+|\slashed\nabla\psi|^2\right)\, \vol_r\\
\lesssim &\underbrace{\int\limits_{\mathcal H^+\cap\{v_+\geq \V\}}  \left((\rd_{v_+}\psi)^2+|\slashed\nabla\psi|^2\right)\,r^2\sin^2\theta\, dv_+\,d\theta\,d\varphi}_{=:I[\V]}\\
&+\underbrace{\int\limits_{\{v_+=\V\} \cap\{r\geq r_0\}} \left(\f{1}{(-\Delta_{RN})^2}(\rd_{v_-}\psi)^2+|\slashed\nabla\psi|^2\right)\,(-\Delta_{RN}) r^2\sin^2\theta\, dv_-\,d\theta\,d\varphi}_{=:II[\V]}.
\end{split}
\end{equation}
Here, $\vol_r$ is chosen such that $\vol= dr\wedge \vol_r$.

This estimate can be achieved using the identity \eqref{EE.intro} with a combination of well-chosen vector fields, namely $\f{1}{(-\Delta_{RN})}\rd_{v_-}+(1-{\eta}^{-1}(-\Delta_{RN}))\rd_{v_+}$ very near $r=r_+$ and $e^{\lambda r}(\rd_{v_+}+\rd_{v_-})$ in the remaining region for $\lambda$ and {$\eta^{-1}$} suitably large.

By assumption i) of Theorem \ref{main.theorem.RN}, $I[\V]$ in \eqref{RN.step1} has the desired decay $\V^{-q}$. To deal with the term $II$, we notice that $\int_{\V}^{+\infty} II[v_+]\, dv_+\lesssim IE[\psi;\V]$. Using this and a standard argument based on the pigeonhole principle (see Section \ref{SecPutTogether}) then gives $II[\V] \lesssim \V^{-q}$. In particular, \eqref{RN.step1} now implies firstly the desired bound for $IE_1[\psi;\V]$ in \eqref{ILED.RN} and secondly we also get the bound 
\begin{equation}\label{RN.bound.on.r0}
\int\limits_{\{v_+\geq V\}\cap\{r=r_0\}} \left((\rd_{v_+}\psi)^2+(\rd_{v_-}\psi)^2+|\slashed\nabla\psi|^2\right)\, \vol_r\lesssim \V^{-q}.
\end{equation}

\underline{Step Two (Section \ref{sec.ILED.1st})}  
In this step, our goal is to obtain stability estimates to the future of $\{r=r_0\}$. To this end, we again rely on the identity \eqref{EE.intro}. To the future of $\{r=r_0\}$ we use the vector field $(1+{\eta}^{-1}(-\Delta_{RN})^{1-\alpha})(\rd_{v_-}+\rd_{v_+})$ very near $r=r_-$ and the vector field $e^{\lambda r}(\rd_{v_+}+\rd_{v_-})$ in the remaining region and choose $\lambda$ and $\eta^{-1}$ to be suitably large. As a consequence, for any $u_1$ and $v_1$ such that $v_1+u_1=2r_0^*$ (where $r_0^*$ is the value of $r^*$ when $r=r_0$), we obtain the following estimate:
\begin{equation}\label{RN.step2}
\begin{split}
&\int\limits_{\substack{\{v_+ \geq v_1 \} \cap \{v_-\leq u_1\}\\ \cap \{ r\leq r_0\} }} \Big( \frac{1}{(-\Delta_{RN})^\alpha} \big[(\rd_{v_+} \psi)^2 + (\rd_{v_-} \psi)^2\big] + |\slashed\nabla \psi|^2 \Big)\,  \vol\\
\lesssim &\int\limits_{\{v_+\geq v_1\}\cap\{r=r_0\}} \left((\rd_{v_+}\psi)^2+(\rd_{v_-}\psi)^2+|\slashed\nabla\psi|^2\right)\, \vol_r\lesssim (v_1)^{-q},
\end{split}
\end{equation}
where the final bound follows from \eqref{RN.bound.on.r0}. Starting from the estimate \eqref{RN.step2}, we make two observations. Firstly, due to the choice of $f_{\gamma_\sigma}$, we can choose $v_1 < \V$ with $\V \lesssim v_1$ (where the implicit constant is independent of $\V$) such that the inclusion 
$$\{v_+ \geq \V\} \cap\{f_{\gamma_\sigma}\leq 1\} \cap \{ r\leq r_0\}\subset \{v_+ \geq v_1 \} \cap \{v_-\leq u_1\} \cap \{ r\leq r_0\}$$holds. This observation and \eqref{RN.step2} then imply the bound for $IE_2[\psi;V]$ in \eqref{ILED.RN}, which, together with the estimates in Step One, imply \eqref{ILED.RN}.

Secondly, using \eqref{RN.step2}, we obtain the following for any $u_1\in \mathbb R$:
\begin{equation}
\label{RoughILED.RN}
\int\limits_{\{ r\leq r_0\} \cap \{v_- \leq u_1\}} \Big( \frac{1}{(-\Delta_{RN})^\alpha} \big[(\rd_{v_+} \psi)^2 + (\rd_{v_-} \psi)^2\big] + |\slashed\nabla \psi|^2 \Big)\, \vol \leq C \;.
\end{equation}

Notice that to go from \eqref{ILED.RN} and \eqref{RoughILED.RN} to \eqref{ImprovedILED.RN}, it remains to improve the bounds for $|\slashed\nabla\psi|$ since $\rd_{v_+}\psi$ and $\rd_{v_-}\psi$ have already been shown to obey even stronger estimates. 

\underline{Step Three (Section \ref{ILED2})}
In order to improve the bounds for $|\slashed\nabla\psi|$, we need an auxiliary estimate. Using the spherical symmetry\footnote{Of course, the Kerr spacetime is not spherically symmetric and thus requires a modification of this part of the argument. As it turns out, one can define differential operators $\Omega^i_{\mathcal H^+}$ and $\Omega_{\CH}^i$ such that while they do not commute with $\Box_g$, the commutators are well-behaved near the horizons and can be controlled, see Section \ref{ILED2}.} of Reissner--Nordstr\"om, we can commute the wave equation with $\Omega^i$ so that $\Omega^i\psi$ is also a solution to the linear wave equation. Therefore, we can apply the estimate in \eqref{RoughILED.RN} for $\Omega^i\psi$ together with the assumption iii) in Theorem \ref{main.theorem.RN} to get
$$\sum_{i=1}^3 \int\limits_{\{ r\leq r_0\} \cap \{v_- \leq u_1\}} \Big( \frac{1}{(-\Delta_{RN})^\alpha} \big[(\rd_{v_+} \Omega^i\psi)^2 + (\rd_{v_-} \Omega^i\psi)^2\big] + |\slashed\nabla \Omega^i\psi|^2 \Big)\, \vol \leq C \;.$$

\underline{Step Four (Section \ref{sec.IILED})} The estimates from Step Three controls $\rd_{v_-}\Omega^i\psi$ with the desired weight in $(-\Delta)$. Recalling that $|\slashed\nabla\psi|^2\sim \sum_{i=1}^3(\Omega^i\psi)^2$, the desired bounds from $|\slashed\nabla\psi|$ thus follows from a Hardy inequality.

\end{proof}

\subsubsection{Instability estimates}

In the two steps of the instability estimates we deal with the regions to the past and to the future of $\gamma_\sigma$ respectively (recall the definition of $\gamma_\sigma$ in \eqref{gamma.sigma.RN.def} and see Figure \ref{FigLastStep} for a depiction of the spacetime regions in the Kerr case). In the first step (see Proposition \ref{RN.lower.bound} below and Section \ref{sec.instability.1} for the Kerr case),  we use the conservation law associated to the vector field\footnote{In the Kerr case, we will use the conservation law associated to $T_{\CH}$, see Section \ref{sec.instability.1}.} $\rd_t$ to prove a lower bound of the energy on $\gamma_\sigma$. In the second step (see Proposition \ref{prop.instability.2.RN} below and Section \ref{sec.instability.2} for the Kerr case), we then propagate the lower bound on $\gamma_\sigma$ to the $\{v_-=u_0\}$ hypersurface. In both of these steps, the stability estimates that have been derived play an important role.

\begin{proposition}\label{RN.lower.bound}
There exists a sequence $v_k\in \mathbb R$ with $v_k\to \infty$ such that the following estimate holds on the hypersurface ${\gamma_\sigma}$:
$$\int\limits_{\gamma_\sigma\cap\{v_+\geq v_k\}} \T(\rd_t,(-df_{\gamma_\sigma})^\sharp)\, \vol_{f_{\gamma_\sigma}} \gtrsim (v_k)^{-(q+\delta)},$$
where $\vol_{f_{\gamma_\sigma}}$ is chosen so that $\vol=df_{\gamma_\sigma}\wedge \vol_{f_{\gamma_\sigma}}$.
\end{proposition}
\begin{proof}[Sketch of proof] We apply \eqref{EE.intro} with $Z=\rd_t$. Since $\rd_t$ is Killing, we in fact obtain a conservation law, which implies the following lower bound for every $v$:
\begin{equation}
\begin{split}
&\int\limits_{\gamma_\sigma\cap\{v_+\geq v\}} \T(\rd_t,(-df_{\gamma_\sigma})^\sharp)\, \vol_{f_{\gamma_\sigma}}\\
\gtrsim &\underbrace{\int\limits_{\mathcal H^+\cap\{v_+\geq v\}} (\rd_{v_+}\psi)^2\, r^2\sin^2\theta\,dv_+\,d\theta\,d\varphi}_{=:I}\\
&-\underbrace{C\int\limits_{\{v_+=v\}\cap\{f_{\gamma_\sigma}\leq 1\}} \left((\rd_{v_-}\psi)^2+(-\Delta_{RN})|\slashed\nabla\psi|^2\right)\,r^2\sin^2\theta\,dv_-\,d\theta\,d\varphi}_{=:II}.
\end{split}
\end{equation}
Notice again that it is important that this is derived from a conservation law and there are no bulk terms. By assumption ii) of Theorem \ref{main.theorem.RN}, the term $I$ is bounded below by $v^{-(q+\delta)}$. To treat the terms II, we crucially rely on the stability estimate \eqref{ILED.RN} which allows us to pick, using the pigeonhole principle, a sequence $v_k\to \infty$ as $k\to\infty$ such that the corresponding term decays as $(v_k)^{-(q+1)+C\sigma(1-\alpha)}$. Since $\delta< 1$, for every $\sigma$, one can choose $\alpha$ close to $1$ to conclude the proof.
\end{proof}
Finally, we prove the main conclusion of Theorem \ref{main.theorem.RN}:
\begin{proposition}\label{prop.instability.2.RN}
For every $u_0\in \mathbb R$, there exists a sequence $v_k\in \mathbb R$ with $v_k\to \infty$ as $k\to \infty$ such that
$$\int\limits_{\{v_-=u_0\}\cap\{v_{+}\geq v_k\}} (\rd_{v_+}\psi)^2 r^2\sin^2\theta\, dv\,d\theta\,d\varphi \gtrsim (v_k)^{-(q+\delta)}.$$
\end{proposition}
\begin{proof}[Sketch of proof]
We use \eqref{EE.intro} with $Z=\rd_{v_+}$ in the region to the future of $\gamma_\sigma$ and to the past of $\{v_-=u_0\}\cup \mathcal C\mathcal H^+$. Noticing that the term ``at the Cauchy horizon'' vanishes\footnote{To justify this, one needs to use \eqref{ILED.RN} and an approximation argument, see the proof of \eqref{EnergyEste4}.} and dropping a boundary term with a good sign, we have
\begin{equation}\label{final.RN}
\begin{split}
&\int\limits_{\{v_-=u_0\}\cap\{v_+\geq v_k\}} (\rd_{v_+}\psi)^2 r^2\sin^2\theta\, dv\,d\theta\,d\varphi\\
\gtrsim &\underbrace{\int\limits_{\gamma_\sigma\cap\{v_+\geq v_k\}} \T(\rd_{v_+},(-df_{\gamma_\sigma})^\sharp)\, \mathrm{vol}_{f_{\gamma_\sigma}}}_{=:I}-\underbrace{C\int\limits_{\substack{\{v_-\leq u_0\}\cap\{v_+\geq v_k\} \\ \cap \{f_{\sigma}\geq 1\}}}((\rd_{v_+}\psi)^2+(\rd_{v_-}\psi)^2+|\slashed\nabla\psi|^2)\,\mathrm{vol}}_{=:II} \;.
\end{split}
\end{equation}
First, $I$, which is the main term, can be bounded below by the term in Proposition \ref{RN.lower.bound}. This is because $\rd_t=\rd_{v_+}-\rd_{v_-}$ and the term $\int\limits_{\gamma_\sigma\cap\{v_+\geq v_k\}} \T(\rd_{v_-},(-df_{\gamma_\sigma})^\sharp)\, \vol_{f_{\gamma_\sigma}}$ has a favourable sign for this one-sided bound. As a consequence, $I\gtrsim (v_k)^{-(q+\delta)}$.

Then, in order to control the bulk error term $II$ in \eqref{final.RN}, the key observation is that by choosing $\sigma> 0$ to be sufficiently large, the stability estimate \eqref{ImprovedILED.RN} implies that $II$ decays faster than any polynomial, i.e., for any $p$, it is bounded by $C_p (v_k)^{-p}$. The conclusion therefore follows by considering sufficiently large $v_k$.
\end{proof}




\subsection{Outline of the paper}
We end the introduction with an outline of the remainder of the paper. We will begin by a brief discussion on the geometry of the interior of the Kerr black hole in Section \ref{sec.geometry}. We will also introduce the preliminaries about performing energy estimates in this section. We then give a precise statement of the main theorem in Section \ref{sec.main.theorem}. The proof of the main theorem will then occupy the remainder of the paper: In Section \ref{sec.stability}, we prove the necessary stability estimates; in Section \ref{sec.instability}, we then prove the instability estimates, using in particular the bounds derived in Section \ref{sec.stability}.

\subsection{Acknowledgements}

We are grateful to Mihalis Dafermos and Yakov Shlapentokh-Rothman for making the preprint \cite{DafShla15} available. We would also like to thank an anonymous referee for many helpful comments on a previous version of the manuscript. J. Luk is supported by the NSF Postdoctoral Fellowship DMS-1204493. J. Sbierski would like to thank Magdalene College, Cambridge, for their financial support.

\section{The interior of subextremal Kerr spacetime}\label{sec.geometry}

We consider the standard $(t,r,\theta, \varphi)$ coordinates on the smooth manifold $\mathcal{M} = \R \times (r_-,r_+) \times \mathbb{S}^2$, where $r_-$ and $r_+$ will be defined momentarily. A Lorentzian metric $g$ on $\mathcal{M}$ is defined by
\begin{equation}\label{Kerr.metric.BL}
g = g_{tt} \, dt^2 + g_{t\varphi}\,(dt \otimes d\varphi + d\varphi \otimes dt) + \frac{\rho^2}{\Delta} \, dr^2 + \rho^2 \, d\theta^2 + g_{\varphi \varphi} \, d\varphi^2 \;,
\end{equation}
where
\begin{equation*}
\begin{aligned}
&\rho^2 = r^2 + a^2 \cos^2\theta\;,  \qquad \qquad &&g_{tt} = -1 + \frac{2Mr}{\rho^2}\;{,} \\
&\Delta = r^2 -2Mr + a^2\;,  &&g_{t\varphi} = -\frac{2Mra\sin^2\theta}{\rho^2}\;{,} \\
& &&g_{\varphi \varphi} = \big[ r^2 + a^2 +\frac{2Mra^2 \sin^2\theta}{\rho^2}\big] \sin^2\theta \;.
\end{aligned}
\end{equation*}
Here, $a$ and $M$, which are required to satisfy $0 <  |a| < M$, are constants representing the angular momentum per unit mass and the mass of the black hole, respectively.

We now define $r_- <  r_+$  to be the roots of $\Delta$ and fix a time orientation on the Lorentzian manifold $(\mathcal{M},g)$ by stipulating that $-\partial_r$ is future directed. The time oriented Lorentzian manifold $(\mathcal{M},g)$ is called the \emph{interior of a subextremal Kerr black hole}. Moreover, let us fix an orientation by stipulating that the Lorentzian volume form $\vol = \rho^2 \sin \theta \,dt \wedge dr \wedge d\theta \wedge {d}\varphi$ is positive.

For later reference we note that the inverse metric $g^{-1}$ in the \emph{Boyer--Lindquist coordinates} $(t, \varphi, r, \theta)$ is given by
\begin{equation}
\label{gInverse}
g^{-1} = \begin{pmatrix}
-\frac{g_{\varphi \varphi}}{\Delta \sin^2 \theta} & \frac{g_{t \varphi}}{\Delta \sin^2 \theta} & 0 & 0 \\
\frac{g_{t \varphi}}{ \Delta \sin^2 \theta} & -\frac{g_{tt}}{\Delta \sin^2 \theta} & 0 & 0 \\
0 & 0 & \frac{\Delta}{\rho^2} &0 \\
0 & 0 & 0& \frac{1}{\rho^2}
\end{pmatrix} \;.
\end{equation}

Let $r^*(r)$ be a function on $(r_-,r_+)$ satisfying $\frac{dr^*}{dr} = \frac{r^2 + a^2}{\Delta}$ and $\overline{r}(r)$ a function on $(r_-,r_+)$ satisfying $\frac{d\overline{r}}{dr} = \frac{a}{\Delta}$. We now define the following functions on $\mathcal M$:
\begin{align*}
v_+ := t + r^* \quad &, \qquad \varphi_+ := \varphi + \overline{r} \;,\\
v_- := r^* - t \quad &, \qquad \varphi_- := \varphi - \overline{r} \;.
\end{align*}
Here, to be precise, $\varphi_+$ and $\varphi_-$ are defined modulo $2\pi$.
It is easy to check that $(v_+, \varphi_+, r, \theta)$ and $(v_-, \varphi_-, r, \theta)$ are coordinate systems for $\mathcal M$. The metric $g$ in these coordinates takes the following form:
\begin{equation*}
\begin{split}
g&= g_{tt} \, dv_+^2 + g_{t\varphi} \, \big( dv_+ \otimes d\varphi_+ + d\varphi_+ \otimes dv_+\big) + g_{\varphi \varphi} \, d\varphi^2_+  +\big(dv_+ \otimes dr + dr \otimes dv_+\big) \\[2pt] &\qquad - a\sin^2\theta \, \big( dr \otimes d\varphi_+ + d\varphi_+ \otimes dr\big) + \rho^2 \, d\theta^2 \\[7pt]
&= g_{tt} \, dv_-^2 - g_{t\varphi} \, \big( dv_- \otimes d\varphi_- + d\varphi_- \otimes dv_-\big) + g_{\varphi \varphi} \, d\varphi^2_-  + \big(dv_- \otimes dr + dr \otimes dv_-\big)\\[2pt]
&\qquad + a\sin^2\theta \, \big( dr \otimes d\varphi_- + d\varphi_- \otimes dr\big) + \rho^2 \, d\theta^2 \;.
\end{split}
\end{equation*}
This shows that in each of the above coordinate systems the metric extends in fact analytically to all positive values of $r$. We now attach the following boundaries to the manifold $\mathcal M$: using the coordinate chart $(v_+, \varphi_+, r, \theta)$ we attach the boundary $\R \times \{r = r_+\} \times \mathbb{S}^2$, which we call the \emph{event horizon} and denote with $\Hp$. Moreover, using the coordinate chart $(v_-, \varphi_-, r, \theta)$ we attach the boundary $\R \times \{r = r_-\} \times \mathbb{S}^2$, which we call the \emph{Cauchy horizon} and denote with $\CH$. The resulting manifold with boundary\footnote{We note explicitly that with our convention, $\mathcal H^+$ only consists of one ({future} affine complete) null hypersurface and $\mathcal C\mathcal H^+$ also only consists of one ({past} affine complete) null hypersurface. In particular, the dotted lines in Figure \ref{FigInt} {are} not part of the $\M$.} is denoted with $\M$ and is depicted using a Penrose-style representation in Figure \ref{FigInt}.

\subsection{Hypersurfaces}

Note that $\langle dv_+, dv_+\rangle  = \langle dv_-, dv_-\rangle  = \frac{a^2 \sin^2 \theta}{\rho^2}$, thus showing that for $a > 0$ the level sets of $v_+$ and $v_-$ are timelike hypersurfaces away from the axis. 

We now define the functions $f^+ := v_+ -r + r_+$ and $f^- := v_- - r + r_-$. An easy computation gives 
\begin{equation}
\label{NormF}
\langle df^+, df^+\rangle  = \langle df^-, df^-\rangle  =  \frac{a^2 \sin^2 \theta}{\rho^2} + \frac{\Delta}{\rho^2} - \frac{2(r^2 + a^2)}{\rho^2}
\end{equation}
which shows that the level sets of $f^+$ and $f^-$ are spacelike hypersurfaces. We introduce the notation $\Sigma^+_c := \{f^+ = c\}$ and $\Sigma^-_c := \{f^- = c\}$. Moreover, it is immediate that the hypersurfaces $\Sigma_c := \{ r = c\}$ are spacelike.

We also define the function $f_{\gamma_\sigma} (v_+, v_-) := v_+ + v_- - \sigma \log(v_+)$ for $v_+$ large enough, where $\sigma > 0$,  and compute
\begin{equation*}
\langle df_{\gamma_\sigma}, df_{\gamma_\sigma}\rangle  =  \frac{a^2 \sigma^2 \sin^2 \theta  }{v_+^2 \rho^2} + 4(1 - \frac{\sigma}{v_+}) \frac{(r^2 + a^2 )^2}{\Delta \rho^2} \;.
\end{equation*}
Hence, for $v_+$ large enough the level sets of $f_{\gamma_\sigma}$ are spacelike hypersurfaces (recall that $\Delta < 0$ on $\mathcal{M}$). Let 
\begin{equation}\label{gamma.def}
\gamma_\sigma := f_{\gamma_\sigma}^{-1}(1).
\end{equation}

We define an orientation on the level sets of $r$ (including the horizon $\Hp$) by stipulating that the volume form $\volr$, given by $\vol = -dr \wedge \volr$, is positive. Similarly, we define positive volume forms $\volsp$, $\volsn$, and $\volg$ by $\vol = df^+ \wedge \volsp$, $\vol = df^- \wedge \volsn$, and $\vol = df_{\gamma_\sigma} \wedge \volg$, respectively.
\begin{figure}[h]
  \centering
  \def\svgwidth{9cm}
    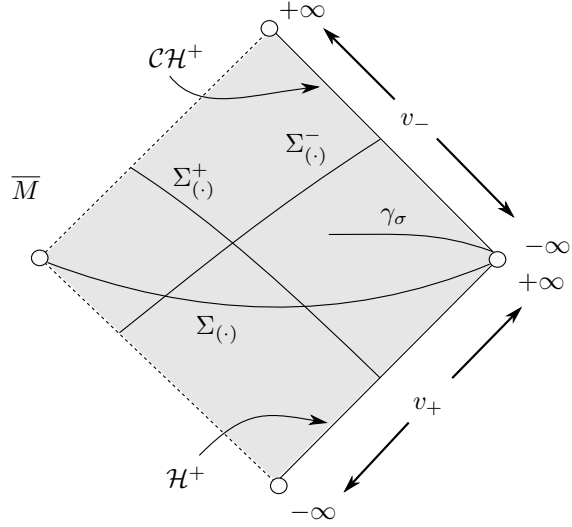
      \caption{The Kerr interior} \label{FigInt}
\end{figure}

\subsection{The principal null frame field}

For convenience we introduce the abbreviations $\mathscr{S} = \sin \theta$ and $\mathscr{C}=\cos \theta$. Moreover, using the Boyer--Lindquist coordinates, we define
\begin{equation*}
V= (r^2 + a^2) \partial_t + a \partial_\varphi  \qquad \textnormal{ and } \qquad W= \partial_\varphi + a \Si^2 \partial_t \;.
\end{equation*}
The \emph{principal null frame} is then given by
\begin{equation*}
\begin{aligned}
&e_1 := \frac{1}{\rho} \partial_\theta \;,  && e_3 := \frac{\Delta}{\rho^2} \partial_r - \frac{1}{\rho^2} V \;,\\
&e_2 := \frac{W}{|W|} = \frac{1}{\rho \Si}(\partial_\varphi + a \Si^2 \partial_t)\;, \qquad \qquad
&&e_4 := -\partial_r - \frac{1}{\Delta}V \;.
\end{aligned}
\end{equation*}
The vector fields $e_3$ and $e_4$ are null and future directed and satisfy $\langle e_3,e_4\rangle  = -2$. Let us denote the distribution spanned by $e_3$ and $e_4$ by $\Pi$ and the to $\Pi$ orthogonal distribution by $\Pi^\perp$. The vector fields $e_1$ and $e_2$ are not defined on the axis, but where defined they form an orthonormal basis for $\Pi^\perp$.  

Note that in $(v_-,r,\theta,\varphi_-)$ coordinates we have\footnote{In the following $\big{|}_\pm$ indicates a partial derivative in the $(v_\pm, r, \theta, \varphi_\pm)$ coordinate system.} 
\begin{equation*}
e_3 = \frac{\Delta}{\rho^2} \frac{\partial}{\partial r} \Big|_- - \frac{2}{\rho^2} V \qquad \textnormal{ and } \qquad e_4 = - \frac{\partial}{\partial r}\Big|_-\;,
\end{equation*}
while in $(v_+, r , \theta, \varphi_+)$ coordinates we have
\begin{equation*}
e_3 = \frac{\Delta}{\rho^2} \frac{\partial}{\partial r}\Big|_+ \qquad \textnormal{ and } \qquad e_4 = - \frac{\partial}{\partial r} \Big|_+ - \frac{2}{\Delta} V \;. 
\end{equation*}
Hence, the null vectors $e_3$ and $e_4$ are regular at the Cauchy horizon $\CH$, but not at the event horizon $\Hp$.
At the event horizon $\Hp$ the vector fields 
\begin{equation*}
\tilde{e}_3 := -\frac{1}{\Delta}\, e_3 \qquad \textnormal{ and } \qquad  \tilde{e}_4 := - \Delta \, e_4
\end{equation*}
are regular.

{Using $\nabla$ to denote the Levi--Civita connection, t}he covariant derivatives can be computed to be
\begin{equation}
\label{CovDer}
\begin{aligned}
\nabla_{e_1} e_1 &= - \frac{r}{2 \rho^2} e_3 + \frac{r \Delta}{2 \rho^4} e_4   \;,
&&\nabla_{e_1} e_2 = -\frac{a\Co}{2 \rho^2} e_3 -\frac{a\Co\Delta}{2 \rho^4} e_4 \;,\\
\nabla_{e_2} e_1 &= \frac{\Co}{\Si} \frac{r^2 + a^2}{\rho^3} e_2 + \frac{a\Co}{2 \rho^2} e_3 + \frac{a\Co\Delta}{2 \rho^4} e_4 \;, \qquad
&&\nabla_{e_2} e_2 = -\frac{\Co}{\Si} \frac{r^2 +a^2}{\rho^3} e_1 - \frac{r}{2\rho^2} e_3 +\frac{r\Delta}{2\rho^4} e_4 \;,\\
\nabla_{e_3} e_1 &= -\frac{\Delta a \Co}{\rho^4} e_2 - \frac{a^2 \Si\Co}{\rho^3} e_3 \;,
&&\nabla_{e_3} e_2 = \frac{a\Co\Delta}{\rho^4} e_1 + \frac{ar \Si}{\rho^3} e_3 \;,\\
\nabla_{e_4} e_1 &= -\frac{a\Co}{\rho^2} e_2 - \frac{a^2 \Si\Co}{\rho^3} e_4 \;,
&&\nabla_{e_4} e_2 = \frac{a\Co}{\rho^2} e_1 - \frac{ar \Si}{\rho^3} e_4\;,
\\[15pt]
\nabla_{e_1} e_3 &= \frac{r \Delta}{\rho^4} e_1 - \frac{\Delta a \Co}{\rho^4} e_2 + \frac{a^2 \Si\Co}{\rho^3} e_3 \;,
&&\nabla_{e_1} e_4 = - \frac{r}{\rho^2} e_1 - \frac{a\Co}{\rho^2} e_2 - \frac{a^2\Si\Co}{\rho^3} e_4 \;,
\\
\nabla_{e_2} e_3 &= \frac{a\Co\Delta}{\rho^4} e_1 + \frac{r\Delta}{\rho^4} e_2 + \frac{ar\Si}{\rho^3} e_3 \;,
&&\nabla_{e_2} e_4 = \frac{a\Co}{\rho^2} e_1 -\frac{r}{\rho^2} e_2 - \frac{ar\Si}{\rho^3} e_4  \;,
\\
\nabla_{e_3} e_3 &= \partial_r\big(\frac{\Delta}{\rho^2}\big) e_3 \;,
&&\nabla_{e_3} e_4 = -\frac{2a^2 \Si\Co}{\rho^3} e_1 +\frac{2ar\Si}{\rho^3} e_2 - \partial_r \big(\frac{\Delta}{\rho^2}\big) e_4 \;,
\\
\nabla_{e_4} e_3 &= - \frac{2a^2 \Si\Co}{\rho^3} e_1 - \frac{2ar\Si}{\rho^3} e_2 \;,
&&\nabla_{e_4} e_4 = 0\;.
\end{aligned}
\end{equation}

The commutators are
\begin{equation*}
\begin{aligned}
&[e_1,e_2] = - \frac{\Co}{\Si} \frac{r^2 + a^2}{\rho^3} \, e_2 {-} \frac{a\Co}{\rho^2} \, e_3 - \frac{a\Co\Delta}{\rho^4} \, e_4 \;,\qquad 
&&[e_2, e_3] = \frac{r\Delta}{\rho^4} \,e_2 \;,
\\
&[e_1, e_3] = \frac{r\Delta}{\rho^4} \, e_1 + \frac{2 a^2 \Si\Co}{\rho^3} \, e_3 \;,
&&[e_2, e_4] = - \frac{r}{\rho^2} \, e_2 \;,
\\
&[e_1, e_4] = - \frac{r}{\rho^2} \, e_1 \;,
&&[e_3,e_4] = \frac{4ar\Si}{\rho^3} \, e_2 - \partial_r \Big( \frac{\Delta}{\rho^2}\Big) \, e_4 \;.
\end{aligned}
\end{equation*}

Finally, we compile the expressions for the Boyer--Lindquist coordinate vector fields written in terms of the principal frame field:
\begin{equation}
\label{CoordVecInFrame}
\begin{aligned}
&\partial_t = - \frac{a\Si}{\rho} \, e_2 - \frac{1}{2} \, e_3 - \frac{\Delta}{2 \rho^2} \, e_4\;, \qquad \qquad &&\partial_\theta = \rho \, e_1 \;,\\
&\partial_r = \frac{\rho^2}{2\Delta} \, e_3 - \frac{1}{2} \, e_4 \;, &&\partial_\varphi = \frac{\Si}{\rho}(r^2 + a^2) \, e_2 + \frac{\Si^2 a}{2} \, e_3 + \frac{\Si^2 a \Delta}{2 \rho^2} \, e_4 \;.
\end{aligned}
\end{equation}

\subsection{Commutators}\label{sec.commutators}

We define $\varphi_{\CH} := \varphi_- {+} \frac{a}{(r_-^2 + a^2)} v_-$, which is a regular angular function away from $\Hp$, and $\varphi_{\Hp} = \varphi_+ - \frac{a}{(r_+^2 + a^2)} v_+$, which is a regular angular function away from $\CH$.
A direct computation gives
\begin{equation}
\label{DerivativesPhi}
\begin{aligned}
e_1 \varphi_{\CH} &= 0 \;, & \quad e_1 \varphi_{\Hp} &=0 \;,\\
e_2 \varphi_{\CH} &= \frac{1}{\rho \Si} (1- \frac{a^2 \Si^2}{r_-^2 + a^2}) \;, & \quad e_2 \varphi_{\Hp} &= \frac{1}{\rho \Si} (1- \frac{a^2 \Si^2}{r_+^2 + a^2})\;, \\
e_3 \varphi_{\CH} &= - \frac{2}{\rho^2} a (1- \frac{r^2 + a^2}{r_-^2 + a^2}) \;, & \quad e_3 \varphi_{\Hp} &=0\;,\\
e_4 \varphi_{\CH} &=0 \;, & \quad e_4 \varphi_{\Hp} &= -\frac{2}{\Delta} a (1- \frac{r^2 + a^2}{r_+^2 + a^2}) \;.
\end{aligned}
\end{equation}

We now define the following vector fields on $\mathcal M$:
\begin{equation}
\label{DefCom}
\begin{aligned}
\Omega_{\CH / \Hp}^1 &= -\frac{\cos \theta \cos \varphi_{\CH / \Hp}}{\sin \theta} W - \sin \varphi_{\CH / \Hp} \partial_\theta \\
 &= - \rho \cos \theta \cos \varphi_{\CH / \Hp}   \cdot e_2 - \rho \sin \varphi_{\CH / \Hp} \cdot  e_1 \;,\\[7pt]
\Omega_{\CH / \Hp}^2 &= - \frac{\sin \varphi_{\CH / \Hp} \cos \theta}{\sin \theta} W + \cos \varphi_{\CH / \Hp} \partial_\theta \\
 &= -\rho \cos \theta \sin \varphi_{\CH / \Hp}   \cdot e_2 + \rho \cos \varphi_{\CH / \Hp} \cdot  e_1 \;,\\[7pt]
\Omega_{\CH / \Hp}^3 &= W\\
&= \rho \sin \theta\cdot  e_2 \;.
\end{aligned}
\end{equation}
In order to understand the regularity properties of these vector fields, {they should be compared to} the generators of the rotations $\Omega^i = \varepsilon_{ijk} x_j \partial_k$ in $\R^3$, which are in particular smooth, {and} read {as follows} in spherical coordinates
\begin{align*}
\Omega^1 &= -\frac{\cos \theta \cos \varphi}{\sin \theta} \,\partial_\varphi - \sin \varphi \, \partial_\theta \;,\\
\Omega^2 &= - \frac{\sin \varphi \cos \theta}{\sin \theta} \,\partial_\varphi + \cos \varphi \, \partial_\theta \;,\\
\Omega^3 &= \partial_\varphi \;.
\end{align*}

It is now easy to see that $\{\Omega_{\CH}^1, \Omega_{\CH}^2, \Omega_{\CH}^3\}$ is a collection of smooth vector fields which span $\Pi^\perp$ everywhere and, moreover, extend smoothly to $\CH$, while $\{\Omega_{\Hp}^1, \Omega^2_{\Hp}, \Omega_{\Hp}^3\}$ is a collection of smooth vector fields which span $\Pi^\perp$ everywhere and extend smoothly to $\Hp$. The angular coordinates $\varphi_{\CH / \Hp}$ have been defined such that one has $[\Omega_{\CH}^i, e_3] = 0$ on $\CH$ and $[\Omega_{\Hp}^i , -\Delta e_4]=0$ on $\Hp$. The importance of this property will become obvious in Section \ref{ILED2}.

\subsection{The wave equation and an energy estimate}\label{sec.energy.estimates}

Let $\psi \in C^\infty(M, \R)$. The \emph{wave equation} on the Kerr interior is defined by
\begin{equation*}
\Box_g \psi := {(g^{-1})}^{\mu \nu} \nabla_{{\mu}} \nabla_\nu \psi =0 \;,
\end{equation*}
where $\nabla${, as above,} denotes the Levi--Civita connection on $(\mathcal{M},g)$ {and here, and below, repeated indices are summed over}. We recall that the \emph{stress-energy tensor} $\T[\psi]$ of $\psi$ is given by
\begin{equation*}
\T[\psi]_{\mu \nu} := \partial_\mu \psi \partial_\nu \psi - \frac{1}{2} g_{\mu \nu} g^{-1}(d\psi, d\psi) 
\end{equation*}
and satisfies $\nabla^\mu \T[\psi]_{\mu \nu} = \Box_g \psi \partial_\nu \psi$. We also recall that the \emph{deformation tensor} $\pi(Z)_{\mu \nu}$ of a vector field $Z$ is given by 
$$\pi(Z)_{\mu \nu} = \nabla_\mu Z_{\nu} + \nabla_\nu {Z}_\mu.$$ 
For a compact region $D \subseteq \M$ with piecewise smooth boundary $\partial D$, which is oriented with respect to the outward pointing normal, Stokes' theorem now yields
\begin{equation}
\label{EnergyEst}
\int\limits_{\partial D} \iota_{\T[\psi](Z, \cdot)^\sharp} \,\vol = \int\limits_{D} d\big(  \iota_{\T[\psi](Z, \cdot)^\sharp} \,\vol\big) = \int\limits_D \Big(\T[\psi]_{\mu \nu} \nabla^\mu Z^\nu + \Box_g \psi Z\psi \Big) \, \vol \;.
\end{equation}
We refer to \eqref{EnergyEst} as the \emph{energy estimate} with multiplier $Z$ in the region $D$. 

\section{The main theorem}\label{sec.main.theorem}

Let 
$$T_{\Hp} = \partial_t + \frac{a}{r_+^2 + a^2} \partial_\varphi$$ denote {a} Hawking vector field of the event horizon; $T_{\Hp}$ is Killing and orthogonal to $\Hp$. Moreover, we denote with 
\begin{equation}\label{TCH.def}
T_{\CH} = (r_-^2 + a^2) \, \partial_t + a \, \partial_\varphi 
\end{equation}
 a Hawking vector field of the Cauchy horizon, which is Killing and orthogonal to $\CH$. {We now state our theorem precisely as follows:}

\begin{theorem}
\label{MainThm}
Let $\psi : \mathcal{M}\cup \Hp \to \R$ be a smooth solution of the wave equation $\Box_g \psi = 0$. Assume that {there exists $q> 0$, $\delta \in [0,1)$ and $C> 0$ such that}
\begin{enumerate}[i)]
\item {f}or all $\V \geq 1${, the following upper bound holds on $\mathcal H^+$:}
\begin{equation*}
\int\limits_{\Hp \cap \{v_+ \geq \V\}} \T[\psi](\tilde{N},\tilde{e}_4)  \,  \volr {\leq C} \V^{-q} \;,
\end{equation*} 
where $\tilde{N}$ is a future directed timelike vector field that satisfies $[\tilde{N},T_{\Hp}] = 0$ { on $\mathcal H^+$;}
\item {f}or all $\V \geq 1${, the following lower bound holds on $\mathcal H^+$:}
\begin{equation}
\label{LowBound}
\V^{-(q + \delta)} {\leq C} \int\limits_{\Hp \cap \{ v_+ \geq \V \}} \T[\psi](T_{\CH}, \tilde{e}_4) \, \volr {;}
\end{equation}
\item {the following upper bound holds for the second order energy\footnote{Recall the definition of $\Omega_{\Hp}^i$ in \eqref{DefCom}.}  on $\mathcal H^+$:}
\begin{equation*}
\sum\limits_{i=1}^3\int\limits_{\Hp \cap \{{v_+} \geq 1\}} \T[\Omega_{\Hp}^i \psi](\tilde{N},\tilde{e}_4) \, \volr \leq C \;.
\end{equation*}
\end{enumerate}
It then follows that for every $u_0 \in \R$ there exists a sequence $v_k \in \R$ with $v_k \to \infty$ for $k \to \infty$ such that the following holds:\footnote{Here, in the preceding and in the following, $\sharp$ denotes the isomorphism between one-forms and vector fields given by ``raising the index with the inverse of the metric $g$''.}
\begin{equation}
\label{PolyLowBound}
\int\limits_{\Sigma^-_{u_0} \cap \{{v_+} \geq v_k\}} \T[\psi]\big(-\Delta e_4, (-df^-)^\sharp\big) \, \volsn \gtrsim (v_k)^{-(q + \delta)} \;,
\end{equation}
{where the implicit constant is independent of $v_k$.}
In particular, \eqref{PolyLowBound} implies\footnote{The fact that \eqref{PolyLowBound} implies \eqref{NotL2} will be proven explicitly in Remark \ref{Rmk.NotL2} below.} that for every $u_0 \in \R$  we have
\begin{equation}
\label{NotL2}
\int\limits_{\Sigma^-_{u_0} \cap \{v_+ \geq 1\}} \T[\psi]\big( e_4, (-df^-)^\sharp\big) \, \volsn = \infty \;.
\end{equation}
\end{theorem}

Note that the right hand side of \eqref{LowBound} is not manifestly non-negative, since $T_{\CH}$ is spacelike on the event horizon $\Hp$. 
\begin{remark}[Alternative formulation of Theorem \ref{MainThm}]
\label{RemarkToThm}
Assumption $ii)$ in Theorem \ref{MainThm} can be replaced by 
\begin{enumerate}[ii')]
\item The wave $\psi$ is axisymmetric {(i.e., $\rd_{\varphi}\psi=0$ everywhere in $\M$)} and there exists a $\delta \in [0,1)$ and a sequence $w_k \in \R$ with $w_k \to \infty$ for $k \to \infty$ such that
\begin{equation}
\label{AltLowBound}
w_k^{-(q + \delta)} \lesssim \int\limits_{\Hp \cap \{ v_+ \geq w_k \}} \T[\psi](T_{\CH}, \tilde{e}_4) \, \volr \;.
\end{equation}
\end{enumerate}
Note that under the assumption of axisymmetry the right hand side of \eqref{AltLowBound} is manifestly non-negative\footnote{This can be seen by noting that $T_{\mathcal H^+}$ is future directed and causal along $\mathcal H^+$ and that for axisymmetric $\psi$, the identity $\mathbb T[\psi](T_{\mathcal C\mathcal H^+},\tilde{e_4})=(r^2_-+a^2)\mathbb T[\psi](T_{\mathcal H^+},\tilde{e_4})$ holds on $\mathcal H^+$.}. Making use of this non-negativity\footnote{See the proof of Proposition \ref{PropLowBoundGamma}.} allows us to weaken the lower bound \eqref{LowBound} {(which holds for every $\V\geq 1$)} to a lower bound only along a sequence $w_k \to \infty$.  
\end{remark}

\begin{remark}[Yet another formulation of Theorem \ref{MainThm}]
By pulling out the weight from under the integral we see that $i)$ follows from
\begin{equation*}
\int_{\Hp \cap \{ v_+ \geq 1\}} (v_+)^q \T[\psi](N,\tilde{e}_4)  \,  \volr <  \infty \;.
\end{equation*}
Moreover, for $w_k:=2^{i_k}$, where $i_k\to \infty$ as $k\to \infty$, $ii')$ follows from:
\begin{equation*}
\int_{\Hp \cap \{ v_+ \geq 1 \}} (v_+)^{q+\delta - \varepsilon} \T[\psi](T_{\CH}, \tilde{e}_4) \, \volr = \infty \qquad \textnormal{ holds for some } \varepsilon> 0\;.
\end{equation*}
This is easily seen by contradiction: Assume that for all $b> 0$ there exists a $k_0 \in \N$ such that for all $k >  k_0$ we have 
\begin{equation*}
\int_{\Hp \cap \{ 2^k \leq v_+ \leq 2^{k+1}\}}\T[\psi](T_{\CH}, \tilde{e}_4) \, \volr <  b \cdot (2^k)^{-(q + \delta)}\;.
\end{equation*} 
It then follows that 
\begin{equation*}
\int_{\Hp \cap \{ 2^k \leq v_+ \leq 2^{k+1}\}} (v_+)^{q + \delta - \varepsilon} \T[\psi](T_{\CH}, \tilde{e}_4) \, \volr <  2^{q + \delta - \varepsilon} b \cdot (2^k)^{-\varepsilon}
\end{equation*}
holds for all $k >  k_0$. Summing over $k$ then gives the contradiction $\int_{\Hp \cap \{ v_+ \geq 1 \}} (v_+)^{q+\delta - \varepsilon} \T[\psi](T_{\CH}, \tilde{e}_4) \, \volr <  \infty$. 
\end{remark}

\begin{remark}[Proof of \eqref{NotL2} from \eqref{PolyLowBound}]\label{Rmk.NotL2}
Recall that $r^*(r)$ satisfies 
\begin{equation*}
\frac{dr^*}{dr} = \frac{r^2 + a^2}{\Delta} = \frac{r^2 + a^2}{(r-r_+)(r-r_-)} = \frac{r_-^2 + a^2}{(r_- - r_+)(r-r_-)} + \mathcal{O}(1) \qquad \textnormal{ for } r \searrow r_- \;.
\end{equation*}
Hence, for any $r_0 \in (r_-, r_+)$ there exists a $C> 0$ such that
\begin{equation*}
\frac{1}{2\kappa_-} \log (r-r_-) + C \geq r^* \geq \frac{1}{2\kappa_-} \log (r-r_-) - C \;,
\end{equation*}
where we have introduced the surface gravity $\kappa_- = \frac{r_- - r_+}{2(r_-^2 + a^2)}$ of the Cauchy horizon. Thus, in $(r_- , r_0)$ we have
\begin{equation}
\label{RCH}
r-r_- \lesssim e^{2 \kappa_- r^*} = e^{\kappa_-(v_+ + v_-)} \lesssim r-r_- \;.
\end{equation}

We now prove \eqref{NotL2} by contradiction, that is we assume
\begin{equation*}
\int\limits_{\Sigma^-_{u_0} \cap \{v_+ \geq 1\}} \T[\psi]\big( e_4,(-df^-)^\sharp\big) \, \volsn \leq C \;.
\end{equation*}
It then follows from \eqref{RCH} and the relation $v_- - r + r_- = u_0$ along $\Sigma^-_{u_0}$ that 
\begin{equation*}
\int\limits_{\Sigma^-_{u_0} \cap \{v_+ \geq \V\}} -\Delta \cdot \T[\psi]\big( e_4,(-df^-)^\sharp\big) \, \volsn \lesssim \int\limits_{\Sigma^-_{u_0} \cap \{v_+ \geq \V\}} e^{\kappa_- v_+} e^{\kappa_-(u_0 + r)} \cdot \T[\psi]\big( e_4,(-df^-)^\sharp\big) \, \volsn \lesssim e^{\kappa_- v_+}   \;.
\end{equation*}
This, however, contradicts \eqref{PolyLowBound}. 
In order to prove Theorem \ref{MainThm} it thus suffices to prove \eqref{PolyLowBound}.
\end{remark}

\begin{remark}[Blow up of non-degenerate energy]\label{rmk.energy}
{Let $\Sigma$ be a smooth spacelike hypersurface that intersects the Cauchy horizon transversally. Suppose $\Sigma$ is given by a defining function $f:\M\to \mathbb R$ (i.e., $\Sigma=f^{-1}(0)$) such that $(-df)^{\sharp}$ is a future-directed timelike vector field. Define the non-degenerate energy by 
$$\int_{\Sigma} \T(e_3+e_4, (-df)^{\sharp})\,\vol_f,$$
where the volume form $\vol_f$ is defined such that $\vol=df\wedge \vol_f$. \eqref{NotL2} implies\footnote{Notice that $
\int\limits_{\Sigma^-_{u_0} \cap \{v_+ \geq 1\}} \T[\psi]\big( e_3, (-df^-)^\sharp\big) \, \volsn > 0 \;.
$} that the non-degenerate energy on $\Sigma^-_{u_0}\cap\{v_+\geq 1\}$ is infinite for every $u_0\in \mathbb R$. 

Moreover, by proving energy estimates {\bf locally near the Cauchy horizon}, it can be shown that the non-degenerate energy is infinite on {\bf any} smooth spacelike hypersurface intersecting the Cauchy horizon transversally, as is claimed in Theorem \ref{main.theorem.intro}.}
\end{remark}

\begin{remark}[Constructing solutions to the wave equation which satisfy the assumptions of Theorem \ref{MainThm}]
It is a standard fact that the linear wave equation is well-posed towards the future with smooth data imposed on the event horizon $\mathcal H^+\cap \{v_+\geq 1\}$ and spacelike hypersurface $\Sigma^+_1$. Since the equation is linear, the solution exists and remains smooth in $\{v_+\geq 1\}\cap\{r\in (r_-,r_+]\}$. Note that the data on $\Hp \cap \{v_+ \geq 1\}$ can be prescribed such that the assumptions $i) - iii)$ of Theorem \ref{MainThm} are satisfied. Moreover, the proof of Theorem \ref{MainThm} will show that the theorem also holds  if $\psi$ is only assumed to be a smooth solution in the smaller set $\{v_+\geq 1\}\cap\{r\in (r_-,r_+]\}$. Hence, there exists a large class of solutions to the wave equation in the interior of the Kerr black hole which are initially regular but have infinite non-degenerate energy near the Cauchy horizon.
\end{remark}

\section{Stability estimates}\label{sec.stability}

\subsection{Integrated energy decay for first derivatives}\label{sec.ILED.1st}

The results of this section depend only on the first assumption of Theorem \ref{MainThm}:
\begin{assumption}
\label{Assumption1}
Let $\psi : \mathcal{M}\cup \Hp \to \R$ be a smooth solution of the wave equation $\Box_g \psi = 0$ and assume that there exists a $q> 0$ such that for all $\V \geq 1$
\begin{equation*}
\int\limits_{\Hp \cap \{v_+ \geq \V\}} \T[\psi](\tilde{N},\tilde{e}_4)  \,  \volr \lesssim \V^{-q} \;,
\end{equation*} 
where $\tilde{N}$ is a future directed timelike vector field that satisfies $[\tilde{N},T_{\Hp}] = 0$.
\end{assumption}


We prove:
\begin{proposition}
\label{PropFirstDerILEDDecay}
Under Assumption \ref{Assumption1} the following holds: Let $\sigma > 0$ and $\V$ be large enough\footnote{This is just to ensure that $f_{\gamma_\sigma}$ is defined in the region $\{f^+ \geq \V\}$}. Then for any $\alpha \in {[}0,1)$ we have 
\begin{equation*}
\begin{split}
\int\limits_{\substack{\{f^+ \geq \V\} \cap \{f_{\gamma_\sigma} \leq 1\} \\ \cap \{ r\geq \rblue\}} } &\Big( (\tilde{e}_4 \psi)^2 + (\tilde{e}_3 \psi)^2 + (e_2 \psi)^2 + (e_1 \psi)^2 \Big) \, \vol \\
&+ \int\limits_{\substack{\{f^+ \geq \V\} \cap \{f_{\gamma_\sigma} \leq 1\} \\ \cap \{ r\leq \rblue\} }} \Big( \frac{1}{(-\Delta)^\alpha} \big[ \Delta^2 (e_4 \psi)^2 + (e_3 \psi)^2\big] + (e_2 \psi)^2 + (e_1 \psi)^2 \Big)\, \vol  \lesssim \V^{-q}\;,
\end{split}
\end{equation*}
for some $\rblue \in (r_-, r_+)$.
\end{proposition} 

\begin{proposition}
\label{PropFirstDerILED}
Under Assumption \ref{Assumption1} the following holds: Let $\alpha \in {[}0,1)$ and $u_1 \in \R$ There then exists a constant $C > 0$  such that
\begin{equation}
\label{ILEDFirst}
\int\limits_{\{ r\leq \rblue\} \cap \{f^- \leq u_1\}} {\Big(}\frac{1}{(-\Delta)^\alpha} \Big[ \Delta^2 {(e_4}\psi{)}^2 + {(e_3}\psi{)}^2\Big] + \big( {(e_2}\psi{)}^2  + {(e_1}\psi{)}^2 \big) \Big)\, \vol \leq C \;
\end{equation}
{for $\rblue \in (r_-, r_+)$ as in Proposition \ref{PropFirstDerILEDDecay}.}
\end{proposition}

{We depict in Figure \ref{FigureStabProps} the regions of spacetime under consideration. The darker shaded region depicts the region of integration in Proposition \ref{PropFirstDerILEDDecay}, which is to the future of $\Sigma^+_V$ and to the past of $\gamma_\sigma$. This region is further divided into two: the future and past of $\{r=r_{blue}\}$, where the weights in the integrated energy decay estimates are different.}\footnote{We remind the readers that for $r_{blue}\in (r_-,r_+)$, $\tilde{e}_4$ and $\tilde{e}_3$ are the regular vector fields in the region to the past of $\{r=r_{blue}\}$ while $e_4$ and $e_3$ are the regular vector fields to the future of $\{r=r_{blue}\}$.} {The lightly shaded region, on the other hand, is the region of integration in Proposition \ref{PropFirstDerILED}. The estimate in Proposition \ref{PropFirstDerILED} is of course most useful to the future of $\gamma_\sigma$, i.e., for $f_{\gamma_\sigma}>  1$, for otherwise Proposition \ref{PropFirstDerILEDDecay} provides a stronger bound}\footnote{at least when $V\geq 1$ is sufficiently large.}{. Notice that in the region $\{f_{\gamma_{\sigma}}> 1\}$, since $(-\Delta)$ is sufficiently small, for later applications, we only need to show that the left hand side of \eqref{ILEDFirst} is bounded.}\footnote{Indeed, it holds that the left hand side of \eqref{ILEDFirst} in Proposition \ref{PropFirstDerILED} decays like $(-u_1)^{-q}$ for $u_1 \to -\infty$. This stronger statement, however, is not needed in what follows.}
\begin{figure}[h]
  \centering
  \def\svgwidth{12cm}
    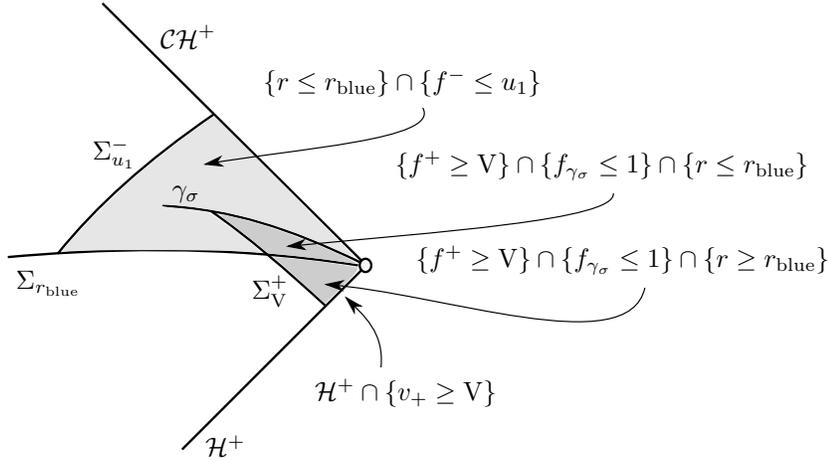
      \caption{The spacetime regions under consideration in Propositions \ref{PropFirstDerILEDDecay} and \ref{PropFirstDerILED}} \label{FigureStabProps}
\end{figure}

{The proof of Propositions \ref{PropFirstDerILEDDecay} and \ref{PropFirstDerILED} will be carried out using the energy estimate \eqref{EnergyEst} with appropriate choices of multipliers $Z$. We will choose three different multipliers in three regions of spacetime - the multipliers will be defined and discussed in Sections \ref{MultiRed}, \ref{MultiBlue} and \ref{MultiInter}. We then conclude this subsection and prove Propositions \ref{PropFirstDerILEDDecay} and \ref{PropFirstDerILED} by combining the energy estimates obtained using these multipliers. }

\subsubsection{Multiplier in the red-shift region}
\label{MultiRed}

The vector field $T_{\Hp} =  \partial_t + \frac{a}{r_+^2 + a^2} \partial_\varphi$ is Killing and orthogonal to the null hypersurface $\Hp$. Moreover, on the event horizon we have
\begin{equation*}
\nabla_{T_{\Hp}} T_{\Hp} = \kappa_+ T_{\Hp} \;,
\end{equation*}
where $\kappa_+ = \frac{r_+ - r_-}{2(r_+^2 + a^2)} > 0$ is the surface gravity. On the event horizon we have $T_{\Hp} = \frac{1}{2(r_+^2 + a^2)} \tilde{e}_4$.

Section 7 in \cite{DafRod08} shows that one can choose an $r_{\mathrm{red}} \in (r_-, r_+)$ and a constant $\kappa> 0$ (depending on $\kappa_+$ and going to zero for $\kappa_+ \to 0$) such that one can construct a {future directed} timelike vector field $N$ that is invariant under the flow of $T_{\Hp}$, i.e., $[N,T_{\Hp}] = 0$, and that satisfies in the region $r \in [\rred, r_+]$
\begin{equation}\label{RedLB}
\T(\psi)_{\mu \nu} \pi({N})^{\mu \nu} \geq \kappa \Big(\big(\tilde{e}_3 \psi\big)^2 + \big(\tilde{e}_4 \psi\big)^2 + \big(e_1 \psi\big)^2 + \big(e_2 \psi\big)^2\Big)\;.
\end{equation}

\subsubsection{Multiplier in the blue-shift region}
\label{MultiBlue}

We define
\begin{equation*}
L:= -\frac{\Delta}{\rho^2} \,e_4 \qquad \textnormal{ and } \qquad \underline{L} := e_3 \;.
\end{equation*}
Let us also write $e_\mu \odot e_\nu := \frac{1}{2}(e_\mu \otimes e_\nu + e_\nu \otimes e_\mu)$
Using \eqref{CovDer} we find that the deformation tensors of $L$ and $\underline{L}$ are given by
\begin{align}
\pi(\underline{L}) &= \frac{2r \Delta}{\rho^4} \,e_1 \odot e_1 + \frac{2 r\Delta}{\rho^4} \, e_2 \odot e_2 +\frac{4a^2\Si\Co}{\rho^3} \, e_1 \odot e_3 + \frac{4ar\Si}{\rho^3} \, e_2 \odot e_3 - \partial_r \Big(\frac{\Delta}{\rho^2}\Big) \,e_4 \odot e_3 \;{,}\notag \\
\pi(L) &= \frac{2 r \Delta}{\rho^4} \, e_1 \odot e_1 + \frac{2 r \Delta}{\rho^4} \, e_2 \odot e_2 - \frac{4 \Delta a^2 \Si\Co}{\rho^5} \, e_1 \odot e_4 + \frac{4\Delta a r\Si}{\rho^5} \, e_2 \odot e_4 - \partial_r \Big(\frac{\Delta}{\rho^2}\Big) \, e_3 \odot e_4 \;. \label{DefL}
\end{align}

Consider now the vector field 
\begin{equation*}
X:= L + \underline{L} = 2\frac{\Delta}{\rho^2} \partial_r = 2(dr)^\sharp \;.
\end{equation*}
Introducing the shorthand $\psi_i := e_i (\psi)$ for $i= 1, \ldots, 4$, we obtain
\begin{equation}
\label{BulkX}
\T[\psi]_{\mu \nu} \pi(X)^{\mu \nu} = 2 \partial_r\Big( \frac{-\Delta}{\rho^2}\Big) [ \psi_1^2 + \psi_2^2] +  \frac{4a^2 \Si \Co}{\rho^3} \psi_1 \psi_3 + \frac{4ar\Si}{\rho^3} \psi_2 \psi_3 - \frac{4a^2 \Delta \Si \Co}{\rho^5} \psi_1 \psi_4 + \frac{4 \Delta a r \Si}{\rho^5} \psi_2 \psi_4 + \frac{4r\Delta}{\rho^4}\psi_3 \psi_4 \;.
\end{equation}
Note here that
\begin{equation}
\label{ControlAngular}
\partial_r\Big( \frac{-\Delta}{\rho^2}\Big)\Big|_{r=r_-} = \frac{r_+ - r_-}{r_-^2 + a^2 \Co^2} > 0\;,
\end{equation}
and thus the first term on the right hand side of \eqref{BulkX} is positive for $r$ close enough to $r_-$.

We now introduce the function
\begin{equation*}
w_\alpha(r) = \int\limits_{r_-}^r \frac{1}{\big(-\Delta(r')\big)^\alpha} \, dr' \;,
\end{equation*}
where $\alpha \in {[}0,1)$ and $r \in [r_-, r_+]$. Note that since $\alpha < 1$ the function $w_\alpha$ is bounded.

We compute
\begin{equation*}
\begin{split}
\T[\psi]_{\mu \nu} \pi(w_\alpha X)^{\mu \nu} &= w_\alpha \T[\psi]_{\mu \nu} \pi(X)^{\mu \nu} + 2\T[\psi] \big((dw_\alpha)^\sharp, X\big) \\
&= w_\alpha \T[\psi]_{\mu \nu} \pi(X)^{\mu \nu} + \frac{2}{(-\Delta)^\alpha} \T[\psi] \big((dr)^\sharp, X\big) \\
&= w_\alpha \T[\psi]_{\mu \nu} \pi(X)^{\mu \nu} + \frac{1}{(-\Delta)^\alpha} \T[\psi] \big(L + \underline{L}, L + \underline{L}\big) \\
&=  w_\alpha \T[\psi]_{\mu \nu} \pi(X)^{\mu \nu}  + \frac{1}{(-\Delta)^\alpha} \Big[ \frac{\Delta^2}{\rho^4} \psi_4^2 + \psi_3^2 + \frac{-2\Delta}{\rho^2} (\psi_1^2 + \psi_2^2)\Big]{.}
\end{split}
\end{equation*}
It now follows from \eqref{BulkX}, \eqref{ControlAngular} and Cauchy{--}Schwarz that there exists an $\rblue \in (r_-, r_+)$ (depending on $\alpha$) such that the following holds in $\R \times (r_-, \rblue) \times \mathbb{S}^2 \subset M$:
\begin{equation}
\label{MultiplierBlue}
\T[\psi]_{\mu \nu} \pi(w_\alpha X)^{\mu \nu} \gtrsim \frac{1}{(-\Delta)^\alpha} \big[ \frac{\Delta^2}{\rho^4} \psi_4^2 + \psi_3^2 \big] + \big( \psi_1^2 + \psi_2^2{\big)}  \;.
\end{equation}

\subsubsection{Multiplier in the intermediate region}
\label{MultiInter}

Let $\tilde{w}_\lambda(r) = e^{\lambda r}$ and compute
\begin{equation*}
\begin{split}
\T[\psi]_{\mu \nu} \pi(\tilde{w}_\lambda X)^{\mu \nu} &= \tilde{w}_\lambda \T[\psi]_{\mu \nu} \pi(X)^{\mu \nu} + 2\lambda \tilde{w}_\lambda \T[\psi]\big((dr)^\sharp ,X\big) \\
&=  \tilde{w}_\lambda \T[\psi]_{\mu \nu} \pi (X)^{\mu \nu} + \lambda \tilde{w}_\lambda \Big[ \frac{\Delta^2}{\rho^4} \psi_4^2 + \psi_3^2 + \frac{-2\Delta}{\rho^2} (\psi_1^2 + \psi_2^2)\Big] \;.
\end{split}
\end{equation*}
Without loss of generality we have $\rblue <  \rred$. We can now choose $\lambda_0> 0$ big enough such that in the region $ r \in [\rblue, \rred]$ we have
\begin{equation}
\label{MultiplierIntermediate}
\T[\psi]_{\mu \nu} \pi(\tilde{w} X)^{\mu \nu} \gtrsim \psi_4^2 + \psi_3^2 + \psi_2^2 + \psi_1^2 \;,
\end{equation}
where we have set $\tilde{w} := \tilde{w}_{\lambda_0}$.

\subsubsection{Putting everything together}
\label{SecPutTogether}

Let $\sigma > 0$ be given and let $\alpha \in {[}0,1)$. Constructing the multiplier $w_\alpha X$ in the blue-shift region such that \eqref{MultiplierBlue} holds determines $\rblue \in (r_-, r_+)$.  
The energy estimate {\eqref{EnergyEst}} with multiplier $N$ in the region $\{r \geq \rred\} \cap \{v_1 \leq f^+ \leq v_2\}$, see also Figure \ref{FigEnEst}, yields
\begin{equation}
\label{EnergyEstimateRed}
\begin{split}
\int\limits_{\Sigma_{\rred} \cap \{v_1 \leq f^+ \leq v_2\}}  &\T[\psi](N,(dr)^\sharp) \, \volr        + \int\limits_{\Sigma^+_{v_2} \cap \{r \geq \rred\}} \T[\psi]({N}, (-df^+)^\sharp) \, \volsp       \\
&+ \int\limits_{\{r \geq \rred\} \cap \{v_1 \leq f^+ \leq v_2\}} \frac{1}{2}\T[\psi]_{\mu \nu} \pi({N})^{\mu \nu} \, \vol  \\
=& \int\limits_{\Sigma^+_{v_1} \cap \{r \geq \rred\}} \T[\psi]({N}, (-df^+)^\sharp) \, \volsp        + \int\limits_{\Hp \cap \{v_1 \leq v_+ \leq v_2 \}} \T[\psi]({N}, (dr)^\sharp) \, \volr \;.
\end{split}
\end{equation}

\begin{figure}[h]
  \centering
  \def\svgwidth{6cm}
    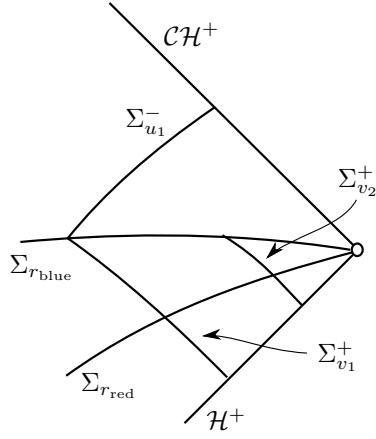
      \caption{Regions considered for the energy estimates in Section \ref{SecPutTogether}} \label{FigEnEst}
\end{figure}
We now note that since $[N, T_{\Hp}] = 0 = [\tilde{N},T_{\Hp}]$, the energy densities $\T[\psi](N, \tilde{e}_4)$ and $\T[\psi](\tilde{N}, \tilde{e}_4)$ are comparable along $\Hp$. The same holds for the energy densities $\T[\psi](N,(dr)^\sharp )$ and $\T[\psi](\tilde{w} X, (dr)^\sharp)$ along $\Sigma_{\rred}$. The energy estimate with multiplier $\tilde{w} X$ in the region $\{\rblue \leq r \leq \rred\} \cap \{v_1 \leq f^+ \leq v_2\}$, together with \eqref{EnergyEstimateRed} and Assumption \ref{Assumption1}, yields
\begin{equation}
\label{UpToRBlue1}
\begin{split}
\int\limits_{\Sigma_{\rblue} \cap \{v_1 \leq f^+ \leq v_2\}} &\T[\psi](\tilde{w} X, (dr)^\sharp) \, \volr + \int\limits_{\Sigma^+_{v_2} \cap \{r \geq \rblue\}} \big[ (\tilde{e}_4 \psi)^2 + (\tilde{e}_3\psi)^2 + (e_{{2}}\psi)^2 + (e_1 \psi)^2 \big] \, \volsp  \\ &\qquad+ \int\limits_{\{r \geq \rblue\} \cap \{v_1 \leq f^+ \leq v_2\}} \big[ (\tilde{e}_4 \psi)^2 + (\tilde{e}_3\psi)^2 + (e_{{2}}\psi)^2 + (e_1 \psi)^2 \big] \, \vol \\
&\lesssim \int\limits_{\Sigma^+_{v_1} \cap \{r \geq \rblue\}} \big[ (\tilde{e}_4 \psi)^2 + (\tilde{e}_3\psi)^2 + (e_{{2}}\psi)^2 + (e_1 \psi)^2 \big] \, \volsp + (v_1)^{-q} \;.
\end{split}
\end{equation}

We need the following
\begin{lemma}
\label{DecayLemma}
Let $q > 0$ and let $h$ be a positive measurable function on $[1,\infty)$ that satisfies
\begin{equation}
\label{PropH}
h(t_2) + \int\limits_{t_1}^{t_2} h(t) \, dt \leq C\cdot \big( h(t_1) + t_1^{-q}\big) \qquad \textnormal{ for all } 1 \leq t_1 <  t_2 <  \infty \;,
\end{equation}
where $C> 0$ is a constant. It then follows that
\begin{equation}
h(t) \lesssim t^{-q} 
\end{equation}
{for every $t\geq 1$, with an implicit constant}\footnote{Indeed, it follows from the proof that the constant depends only on $C$, $q$ and $h(1)$.} {which is independent of $t$.}
\end{lemma}

\begin{proof}
We start with showing
\begin{equation}
\label{qInteger}
h(t) \lesssim t^{-k}
\end{equation}
for all $0 \leq k \leq \left \lfloor{q}\right \rfloor$.

The case $k=0$ follows directly from \eqref{PropH} {after taking $t_1=1$ and $t_2=t$}. So assume \eqref{qInteger} holds for a { non-negative integer} $k \leq \left \lfloor{q}\right \rfloor - 1$. We will prove \eqref{qInteger} for {$k$ replaced by} $k+1$.

Let $\tau_n := 2^n$. In particular, \eqref{PropH} together with \eqref{qInteger} imply \begin{equation*}
\int_{\tau_n}^{\tau_{n+1}} h(t) \, dt \leq C_k \tau_n^{-k} \quad \textnormal{ for all } n \in \N \;,
\end{equation*}
where $C_k > 0$ is a constant\footnote{We will continue to write $C_k$ below to emphasize the dependence of the constant on $k$. Notice however that we will allow the constants to be different in every line.} depending on $k${, $h(1)$, $q$ and $C$}. Thus, for every $n \in \N$ there exists a $\tau'_n \in [\tau_n, \tau_{n+1}]$ such that 
\begin{equation}
\label{SeqMoreDecay}
h(\tau_n') \leq C_k \tau_n^{-(k+1)}
\end{equation} 
holds. Moreover, it follows from \eqref{PropH} that for all $\tau \in [{\tau_{n+1}, \tau_{n+2}}]$ we have
\begin{align}
h(\tau) &\leq C\big(h(\tau_n') + (\tau_n')^{-q} \big) \notag \\ 
&\overset{\eqref{SeqMoreDecay}}{\leq} C_k \big( \tau_n^{-(k+1)} + \tau_n^{-q} \big)  \label{SecondLastStep} \\ 
&\leq C_k \tau_n^{-(k+1)} \notag \;,
\end{align}
which proves \eqref{qInteger} for {$k$ replaced by $k+1$ if} $k+1 \leq \left \lfloor{q}\right \rfloor$.

{The above induction thus allows us to show \eqref{qInteger} for $k\leq \lfloor{q}\rfloor$.} {To continue, we let $k=\lfloor{q}\rfloor$ and} repeat the above proof up to \eqref{SecondLastStep}. From there we conclude $h(\tau) \leq C_k \tau_n^{-q}$, which proves the lemma.
\end{proof}

We now set 
\begin{equation*}
h(t) := \int\limits_{\Sigma^+_{t} \cap \{r \geq \rblue\}} \big[ (\tilde{e}_4 \psi)^2 + (\tilde{e}_3\psi)^2 + (e_{{2}}\psi)^2 + (e_1 \psi)^2 \big] \, \volsp
\end{equation*}
and recall $\vol = df^+ \wedge \volsp$. Hence, \eqref{UpToRBlue1} implies that $h(t)$ satisfies \eqref{PropH} and thus, by Lemma \ref{DecayLemma}, we have $h(t) \lesssim t^{-q}$. Using this in \eqref{UpToRBlue1} and letting $v_2$ tend to infinity gives
\begin{equation}
\label{DecayUpToRBlue}
\int\limits_{\Sigma_{\rblue} \cap \{v_1 \leq f^+\}} \T[\psi](\tilde{w} X, (dr)^\sharp) \, \volr + \int\limits_{\{r \geq \rblue\} \cap \{v_1 \leq f^+ \}} \big[ (\tilde{e}_4 \psi)^2 + (\tilde{e}_3\psi)^2 + (e_1\psi)^2 + (e_1 \psi)^2 \big] \, \vol \lesssim
 (v_1)^{-q} \;.
\end{equation}

Finally, let us consider the energy estimate with multiplier $w_\alpha X$ in the blue-shift region $\{r \leq \rblue\} \cap \{f^- \leq u_1\}$, where $u_1 = 2\rblue^* - 2\rblue - v_1 + r_+ + r_-$ (see also Figure \ref{FigEnEst}):\footnote{Actually, one would carry out the energy estimate first in the compact region $\{r \leq \rblue\} \cap \{f^- \leq u_1\} \cap \{f^+ \leq v_2\}$. The boundary term along $\Sigma^+_{v_2}$ has a positive sign and can thus be dropped. One then takes $v_2 \to \infty$. This approximation argument is standard and will be silently omitted in the future.}
\begin{equation}
\label{ILEDCH}
\begin{split}
\int\limits_{\Sigma^-_{u_1} \cap \{r \leq \rblue\}} &\T[\psi](w_\alpha X, (-df^-)^\sharp) \, \volsn + \int\limits_{\{r \leq \rblue\} \cap \{f^- \leq u_1\}} {\Big(}\frac{1}{(-\Delta)^\alpha} \Big( \Delta^2 \psi_4^2 + \psi_3^2\Big) + \big( \psi_1^2 + \psi_2^2\big) \Big) \, \vol \\
&\lesssim \int\limits_{\Sigma_{\rblue} \cap \{v_1 \leq f^+\}} \T[\psi](w_\alpha X, (dr)^\sharp) \, \volr \;.
\end{split}
\end{equation} 
This, together with \eqref{DecayUpToRBlue} and the fact that the energy densities $\T[\psi](\tilde{w} X, (dr)^\sharp)$ and $\T[\psi](w_\alpha X, (dr)^\sharp)$  are comparable along $\Sigma_{\rblue}$, proves in particular Proposition \ref{PropFirstDerILED}.

In order to prove Proposition \ref{PropFirstDerILEDDecay} let $\V := f^+\big(\gamma_\sigma \cap \Sigma^-_{u_1}\big)$ (this is well-defined for $v_1$ big enough - see also Figure \ref{FigShiftGamma}). Note that ${v_+}\big(\gamma_\sigma \cap \Sigma^-_{u_1}\big)$ is implicitly given by
\begin{equation*}
{v_+}\big(\gamma_\sigma \cap \Sigma^-_{u_1}\big) + u_1 + r(\gamma_\sigma \cap \Sigma^-_{u_1}) - r_- - \sigma \log\big({v_+}\big(\gamma_\sigma \cap \Sigma^-_{u_1}\big)\big) = 1 \;,
\end{equation*}
and thus, together with $u_1 = 2\rblue^* - 2\rblue - v_1 + r_+ + r_-$, we obtain $\V \lesssim {v_+}\big(\gamma_\sigma \cap \Sigma^-_{u_1}\big) \lesssim v_1$. 
\begin{figure}[h]
  \centering
  \def\svgwidth{6cm}
    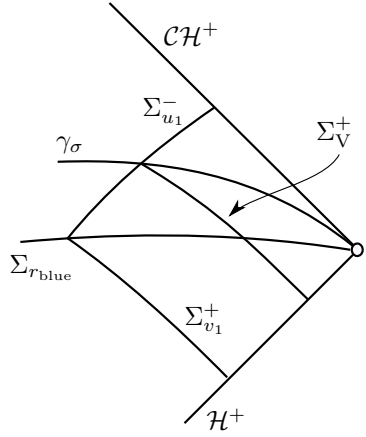
      \caption{The hypersurfaces $\Sigma^+_{v_1}$ and $\Sigma^+_{\V}$}
      \label{FigShiftGamma}
\end{figure}
Thus, \eqref{DecayUpToRBlue} and \eqref{ILEDCH} imply
\begin{equation*}
\begin{split}
\int\limits_{\{r\geq \rblue \} \cap \{f^+ \geq \V\}} &\big[ (\tilde{e}_4 \psi)^2 + (\tilde{e}_3\psi)^2 + (e_{{2}}\psi)^2 + (e_1 \psi)^2 \big] \, \vol \\ 
&\qquad+ \int\limits_{\{r \leq \rblue\} \cap \{f_{\gamma_\sigma} \leq 1\} \cap \{f^+ \geq \V\}} {\Big(}\frac{1}{(-\Delta)^\alpha} {\Big[} \Delta^2 \psi_4^2 + \psi_3^2{\Big]} + \big( \psi_1^2 + \psi_2^2{\big)} {\Big)}\, \vol \\
&\leq\int\limits_{\{r \geq \rblue\} \cap \{ f^+ \geq v_1 \}} {\Big(}\frac{1}{(-\Delta)^\alpha} \Big( \Delta^2 \psi_4^2 + \psi_3^2\Big) + \big( \psi_1^2 + \psi_2^2\big) \Big) \, \vol \\
&\qquad+  \int\limits_{\{r \leq \rblue\} \cap \{f_{\gamma_\sigma} \leq 1\} \cap \{f^- \leq u_1\}} {\Big(}\frac{1}{(-\Delta)^\alpha} \Big( \Delta^2 \psi_4^2 + \psi_3^2\Big) + \big( \psi_1^2 + \psi_2^2\big) \Big)\, \vol \\
&\lesssim (v_1)^{-q}  \lesssim \V^{-q} \;,
\end{split}
\end{equation*}
which proves Proposition \ref{PropFirstDerILEDDecay}.

\subsection{Integrated energy decay for second derivatives}
\label{ILED2}

The result of this section depends, in addition to Assumption \ref{Assumption1}, also on the third assumption of Theorem \ref{MainThm}:
\begin{assumption}
\label{Assumption2}
Assume $\psi : \mathcal{M}\cup \Hp \to \R$ is a smooth solution to the wave equation $\Box_g \psi = 0$ that satisfies
\begin{equation*}
\sum\limits_{i=1}^3\int\limits_{\Hp \cap \{{v_+} \geq 1\}} \T[\Omega_{\Hp}^i \psi](\tilde{N},\tilde{e}_4) \, \volr \leq C \;.
\end{equation*}
\end{assumption}
We prove
\begin{proposition}
\label{PropILED2}
Under Assumption{s} \ref{Assumption1} and \ref{Assumption2} the following holds: Let $\alpha \in {[}0,1)$ and $u_1 \in \R$. There then exists a constant $C > 0$ such that
\begin{equation}
\label{ILEDSecond}
\sum\limits_{i=1}^3\int\limits_{\{ r\leq \rblue\} \cap \{f^- \leq u_1\}} {\Big(}\frac{1}{(-\Delta)^\alpha} \Big( \Delta^2 (e_4\Omega_{\CH}^i\psi)^2 + (e_3\Omega_{\CH}^i\psi)^2\Big) + \big( (e_1\Omega_{\CH}^i\psi)^2 + (e_2\Omega_{\CH}^i\psi)^2\big) \Big)\, \vol \leq C \;
\end{equation}
{for $\rblue \in (r_-, r_+)$ as in Proposition \ref{PropFirstDerILEDDecay}.}
\end{proposition}

\subsubsection{Boundedness of second order energy on $\Sigma_{\rblue}$}
\label{BoundednessSecOrdEnBlueSec}

For the proof of Proposition \ref{PropILED2} we commute the wave equation with the vector fields $\Omega^i_{\Hp, \CH}$. For a vector field $Z$ one has\footnote{See for example \cite{Alin10}, Chapter 6.2}
\begin{equation}
\label{Commutator}
[\Box_g, Z] \psi = \pi(Z)^{\mu \nu} \nabla_\mu \nabla_\nu \psi + {(}\nabla_\mu \pi(Z)^{\mu \nu} {)}\partial_\nu \psi - \frac{1}{2} \partial^\mu \big(\mathrm{tr} \pi(Z)\big) \partial_\mu \psi \;.
\end{equation}
Summing the energy estimates for the functions $\Omega^i_{\Hp} \psi$ in the region $\{r' \leq r \leq r_+\} \cap \{f^+ \geq v_1\}$, where $r' \geq \rred$, with multiplier $N$ {(recall the definition of $N$ and the bound \eqref{RedLB} from Section \ref{MultiRed})} gives
\begin{equation}
\label{SecOrdEnEstRed}
\begin{split}
&\sum_{i=1}^3\int\limits_{\Sigma_{r'} \cap \{f^+ \geq v_1\}} \T[\Omega^i_{\Hp} \psi](N, (dr)^\sharp) \, \volr \\ 
&+ \sum\limits_{i=1}^3\int\limits_{\{r \geq r'\} \cap \{f^+ \geq v_1\}} \kappa\Big(  (\tilde{e}_3 \Omega^i_{\Hp} \psi)^2 +  (\tilde{e}_4 \Omega^i_{\Hp} \psi)^2 +(e_1\Omega^i_{\Hp} \psi)^2 +(e_2 \Omega^i_{\Hp} \psi)^2\Big) \, \vol \\
&+ \sum\limits_{i=1}^3\int\limits_{\{r \geq r'\} \cap \{f^+ \geq v_1\}} \underbrace{\Big(\pi(\Omega^i_{\Hp})^{\mu \nu} \nabla_\mu \nabla_\nu \psi + {\big(}\nabla_\mu \pi(\Omega^i_{\Hp})^{\mu \nu} {\big)}\partial_\nu \psi - \frac{1}{2} {\Big(}\partial^\mu \big(\mathrm{tr} \pi(\Omega^i_{\Hp})\big){\Big)} \partial_\mu \psi\Big)}_{=:Err} {\big(}N \Omega^i_{\Hp}\psi{\big)} \, \vol \\
\leq & \sum\limits_{i=1}^3\int\limits_{\Sigma^+_{v_1} \cap \{r \geq r'\}} \T[\Omega^i_{\Hp}\psi](N, (-df^+)^\sharp) \, \volsp + \sum\limits_{i=1}^3 \int\limits_{\Hp \cap \{{v_+} \geq v_1 \}} \T[\Omega^i_{\Hp}\psi](N, (dr)^\sharp) \, \volr \;.
\end{split}
\end{equation}
We want to show that \eqref{SecOrdEnEstRed} implies\footnote{Note that {for any fixed $v_+$,} the two terms on the right hand side of \eqref{SecOrdEnEstRed} are bounded  (uniformly in $r' \geq \rred$) by {the smoothness} assumption {and Assumption \ref{Assumption2}}.}
\begin{equation}
\label{GronwallRed}
\sum_{i=1}^3\int\limits_{\Sigma_{r'} \cap \{f^+ \geq v_1\}} \T[\Omega^i_{\Hp} \psi](N, (dr)^\sharp) \, \volr \lesssim 1 + \sum\limits_{i=1}^3\int\limits_{\{r \geq r'\} \cap \{f^+ \geq v_1\}} \T[\Omega^i_{\Hp} \psi](N, (dr)^\sharp) \, \vol \;.
\end{equation}
This, together with $\vol = -dr \wedge \volr$ and Gronwall, then implies
\begin{equation}
\label{BoundSecOrdEnRed}
\sum_{i=1}^3\int\limits_{\Sigma_{\rred} \cap \{f^+ \geq v_1\}} \T[\Omega^i_{\Hp} \psi](N, (dr)^\sharp) \, \volr \leq C\;.
\end{equation}
In order to show \eqref{GronwallRed}, let us first recall that $(dr)^\sharp = \frac{1}{2}(L + \underline{L}) = \frac{1}{2}(\frac{1}{\rho^2} \tilde{e}_4 - \Delta \tilde{e}_3)$. Hence, we have
\begin{equation}
\label{ControlStressEnergyRed}
\T[\Omega^i_{\Hp} \psi](N, (dr)^\sharp) \gtrsim  -\Delta(\tilde{e}_3 \Omega^i_{\Hp} \psi)^2 +  (\tilde{e}_4 \Omega^i_{\Hp} \psi)^2 +(e_1\Omega^i_{\Hp} \psi)^2 +(e_2 \Omega^i_{\Hp} \psi)^2 \;.
\end{equation}

We now explain how one estimates the third term in \eqref{SecOrdEnEstRed}: First recall that the tilded frame field $\{\tilde{e}_3, \tilde{e}_4, e_1, e_2\}$ is invariant under the flow of the Killing vector field $T_{\Hp}$. However, the orthonormal basis $\{e_1, e_2\}$ of $\Pi^\perp$ is not smooth on the axis. We thus introduce in addition an orthonormal basis $\{e_1', e_2'\}$ of $\Pi^\perp$ which is smooth in a small neighbourhood of the axis and also invariant under the flow of $T_{\Hp}$.\footnote{Note that so far there was no need to introduce a frame that is also regular at the axis, since we only decomposed expressions involving $d\psi$, i.e., \emph{first} derivatives of $\psi$, with respect to a frame field involving $e_1$ and $e_2$. Technically, in such computations one restricts the domain to the region under consideration with the axis deleted. After estimating and rearranging such expressions we always ended up with the expression $(e_1 \psi)^2 + (e_2 \psi)^2$, which is equal to $(g|_{\Pi^\perp})^{-1}(d\psi, d\psi)$ and thus extends smoothly to the axis. By continuity, the obtained estimate thus also holds on the axis.

Now, however, we are about to decompose expressions involving $\nabla \nabla \psi$. Here, it is important that we decompose with respect to a regular frame near the axis, since for example the term $e_2(e_1 \psi)$ is in general unbounded when approaching the axis.}

We then split the domain of integration in the third term of \eqref{SecOrdEnEstRed} in two disjoint regions $A$ (disjoint of the axis) and $A'$ (containing the axis) such that $e_1$ and $e_2$ are smooth on $\overline{A}$ and $e_1'$ and $e_2'$ are smooth on $\overline{A'}$. 
In the following we decompose $Err$ in region $A$ with respect to the frame $\{\tilde{e}_3, \tilde{e}_4, e_1, e_2\}$ and in region $A'$ with respect to the frame $\{\tilde{e}_3, \tilde{e}_4, e_1', e_2'\}$.  To simplify the presentation, we will restrict ourselves in the following discussion to region $A$. Region $A'$, however, is dealt with completely analogously.

From the coordinate representations of $\Omega^i_{\Hp}$ in \eqref{DefCom} it is easy to see that these vector fields are invariant under the flow of the Killing vector field $T_{\Hp}$. Hence, the same holds for the deformation tensor $\pi(\Omega^i_{\Hp})$  and its covariant derivative. The invariance under the flow of $T_{\Hp}$ implies that the coefficients of the deformation tensor and its covariant derivative, when expressed in terms of the tilded frame, are uniformly bounded in the region under consideration.
We proceed by writing all contractions in the term $Err$ in terms of the tilded frame, e.g.\  $\pi(\Omega^i_{\Hp})^{\mu \nu} \nabla_\mu \nabla_\nu \psi = \pi(\Omega^i_{\Hp})^{\tilde{4} \tilde{4}} (\tilde{e}_4 \tilde{e}_4 \psi + \nabla_{\tilde{e}_4} \tilde{e}_4 \psi + \ldots)$. Again by the invariance of the frame field under the flow of the Killing vector field $T_{\Hp}$ vector fields of the form $\nabla_{\tilde{e}_4} \tilde{e}_4$ have uniformly bounded coefficients when expressed in the tilded frame. We conclude that in region $A$ the term $Err$ can be written as a sum of terms of the form $h \cdot e(f\psi)$ and $h \cdot e\psi$, where $h$ is a uniformly bounded function and $e$ and $f$ are members of the tilded frame field. 

Using this structure, we can now estimate in region $A$ terms of the form $(h \cdot e\psi) (N \Omega^i_{\Hp} \psi)$ in the third term of \eqref{SecOrdEnEstRed} by 
\begin{equation}
\label{CauchySchwarz}
|(h \cdot e\psi) (N \Omega^i_{\Hp} \psi)| \leq \varepsilon^{-1} h^2 (e \psi)^2 + \varepsilon ({N} \Omega^i_{\Hp} \psi)^2 \;.
\end{equation}
Choosing $\varepsilon> 0$ small enough and recalling that $N$ is invariant under the flow of $T_{\Hp}$, the second term can be absorbed by the second term in \eqref{SecOrdEnEstRed}. The first term in \eqref{CauchySchwarz}, after integration, is bounded by Proposition \ref{PropFirstDerILEDDecay}. 

It thus remains to deal with terms of the form $h \cdot e(f\psi)$ in $Err$. We again use Cauchy{--}Schwarz to estimate
\begin{equation*}
|\big(h \cdot e(f\psi)\big) (N \Omega^i_{\Hp} \psi)| \leq \varepsilon^{-1} h^2 \big(e (f\psi)\big)^2 + \varepsilon ({N} \Omega^i_{\Hp} \psi)^2 \;.
\end{equation*}
Choosing $\varepsilon> 0$ small enough, the second term can again be absorbed. Recalling \eqref{ControlStressEnergyRed}, we now show that the first term can be controlled by $\sum_{i=1}^3\T[\Omega^i_{\Hp} \psi](N, (dr)^\sharp)$ (modulo first order terms which, after integration, can again be controlled by Proposition \ref{PropFirstDerILEDDecay}). Here, clearly, the structure of $\pi(\Omega^i_{\Hp})$ is important.
\begin{enumerate}
\item First consider terms of the form $h^2 \cdot \big(e(f \psi)\big)^2$, where $f \in \{e_1, e_2\}$ and $e \in \{e_1, e_2, \tilde{e}_4\}$. Note that we can assume without loss of generality this order for $e$ and $f$, since the commutator is a first derivative of $\psi$ and can thus be estimated as before. In region $A$ we can write $e_1$ and $e_2$ as a linear combination of the $\Omega_{\Hp}^i$ where the coefficient functions have \emph{uniformly bounded} derivative (since we are away from the axis). Writing $f$ thus and using the Leibniz rule, we see that the term $h^2 \cdot \big(e(f \psi)\big)^2$ is controlled by $\sum_{i=1}^3\T[\Omega^i_{\Hp} \psi](N, (dr)^\sharp)$ together with, after integration, Proposition \ref{PropFirstDerILEDDecay}.

\item Now consider terms of the form $h^2 \cdot \big(\tilde{e}_3 (f \psi)\big)$, where $f \in \{e_1, e_2\}$. Again, without loss of generality, we can assume this order of the derivatives. We will show that $h$ goes to zero like $-\Delta$ for $r \to r_+$, and hence this term can be controlled as before\footnote{Notice that $(-\Delta)$ weight for the $\tilde{e}_3\psi$ term in the estimate \eqref{ControlStressEnergyRed}.}.

Note that for $\Omega_{\Hp}^i$ and $f = e_j$, $i \in \{1,2,3\}$, $j \in \{1,2\}$, the coefficient $h$ is here given by 
\begin{equation*}
h =2 \pi(\Omega_{\Hp}^i)^{\tilde{3}j} = 2(-\Delta)\pi(\Omega_{\Hp}^i)^{3j}\;.
\end{equation*} 
Using in particular \eqref{DerivativesPhi} and \eqref{CovDer}, we compute\footnote{Notice that after multiplying the following components of the deformation tensor by $(-\Delta)$, they vanish {on} $\Hp$. This fact can be traced back to $[\Omega_{\Hp}^i, \tilde{e}_4] =0$ on $\Hp$.}
\begin{equation*}
\begin{aligned}
&-2 \pi(\Omega_{\Hp}^1)^{31} = \langle \nabla_{e_4} \Omega_{\Hp}^1, e_1\rangle  + \langle \nabla_{e_1} \Omega_{\Hp}^1, e_4\rangle   &&\\
&\qquad \qquad \qquad= 2 a \cos \varphi_\Hp \Big( - \frac{\Co^2}{\rho} + \frac{\rho}{\Delta}\big(1- \frac{r^2 + a^2}{r_+^2 + a^2}\big)\Big)  \;,&&\\
&-2 \pi(\Omega_{\Hp}^1)^{32} = 2 a \Co \sin \varphi_{\Hp}\Big( \frac{1}{\rho} - \frac{\rho}{\Delta} \big( 1 - \frac{r^2 + a^2}{r_+^2 + a^2}\big)\Big)\;, &&\\[5mm]
&-2 \pi(\Omega_{\Hp}^2)^{31} = 2a \sin \varphi_{\Hp} \Big( - \frac{\Co^2}{\rho} + \frac{\rho}{\Delta} \big( 1 - \frac{r^2 + a^2}{r_+^2 + a^2}\big)\Big) \;,\qquad \qquad  &&-2\pi(\Omega^3_{\Hp})^{31} = \frac{2a \Co \Si}{\rho}{\;,} \\
&-2 \pi(\Omega_{\Hp}^2)^{32} = 2a \Co \cos \varphi_{\Hp} \Big( -\frac{1}{\rho} + \frac{\rho}{\Delta} \big( 1 - \frac{r^2 + a^2}{r_+^2 + a^2}\big)\Big) \;,&&-2\pi(\Omega^3_{\Hp})^{32} = 0\;.
\end{aligned}
\end{equation*}
Hence, $h$ goes to zero like $-\Delta$ for $r \to r_+$.

\item Indeed, no terms of the form $h^2 \cdot \big(\tilde{e}_3 (\tilde{e}_3 \psi)\big)^2$ or $h^2 \cdot \big(\tilde{e}_4 (\tilde{e}_4 \psi)\big)^2$ are present. This follows from
\begin{equation*}
\langle \nabla_{e_4} \Omega_{\Hp}^i, e_4\rangle  = 0 = \langle \nabla_{e_3} \Omega_{\Hp}^i, e_3\rangle 
\end{equation*}
for all $i \in \{1,2,3\}$.

\item It remains to consider terms of the form $h^2 \cdot \big(\tilde{e}_4(\tilde{e}_3 \psi)\big)^2$. Expressing the wave equation in the tilded frame field we obtain
\begin{equation*}
0 = \Box_g \psi = - \tilde{e}_4 ( \tilde{e}_3 \psi) + {(}\nabla_{\tilde{e}_4} \tilde{e}_3{)} \psi + \sum_{i=1}^2 \Big( e_i ( e_i \psi) - \nabla_{e_i} e_i \psi \Big) \;.
\end{equation*}
Hence, we can replace $\big(\tilde{e}_4(\tilde{e}_3 \psi)\big)^2$ by terms we already know how to control.
\end{enumerate}
We have now shown \eqref{GronwallRed} and thus proved \eqref{BoundSecOrdEnRed}.

In the next step we consider the summed energy estimates for the functions $\Omega_{\Hp}^i \psi$ in the region $\{\rblue \leq r \leq \rred\} \cap \{f^+ \geq v_1\}$ with multiplier $\tilde{w} X$. {(Recall the definition in Section \ref{MultiInter}.)} Note that in the region $\{\rblue \leq r \leq \rred\}$ we have
\begin{equation*}
\T[\Omega_{\Hp}^i{\psi}](\tilde{w}X, (dr)^\sharp) \gtrsim (\tilde{e}_3 \Omega^i_{\Hp} \psi)^2 +  (\tilde{e}_4 \Omega^i_{\Hp} \psi)^2 +(e_1\Omega^i_{\Hp} \psi)^2 +(e_2 \Omega^i_{\Hp} \psi)^2 \;.
\end{equation*}
Hence, in order to be able to control the error term by  $\T[\Omega_{\Hp}^i{\psi}](\tilde{w}X, (dr)^\sharp)$ (together with Proposition \ref{PropFirstDerILEDDecay}) as before, the only structure needed this time is that no $\tilde{e}_3(\tilde{e}_3 \psi)$ or $\tilde{e}_4(\tilde{e}_4\psi)$ terms appear in $Err$ - which we have already shown. One then obtains, using \eqref{BoundSecOrdEnRed} and Gronwall,
\begin{equation}
\label{BoundSecOrdEnInt}
\sum_{i=1}^3\int\limits_{\Sigma_{\rblue} \cap \{f^+ \geq v_1\}} \T[\Omega^i_{\Hp} \psi](\tilde{w}X, (dr)^\sharp) \, \volr \leq C\;.
\end{equation}

\subsubsection{Integrated energy decay for second derivatives in the blue-shift region}

We sum the energy estimates for the functions $\Omega_{\CH}^i \psi$ in the region $\{r' \leq r \leq \rblue\} \cap \{f^- \leq u_1\}$ with multiplier $w_\alpha X$. {(Recall the definition from Section \ref{MultiBlue}).} Here, $u_1$ is given again by $u_1 = 2r^*_\mathrm{blue} - 2 \rblue - v_1 + r_+ + r_-$. 
\begin{equation}
\label{SecOrdEnergyEstBlue}
\begin{split}
\sum_{i=1}^3&\int\limits_{\Sigma_{r'} \cap \{f^- \leq u_1\}} \T[\Omega^i_{\CH} \psi](w_\alpha X, (dr)^\sharp) \, \volr \\ 
&+ \sum\limits_{i=1}^3\int\limits_{\{r' \leq r \leq \rblue\} \cap \{f^- \leq u_1\}} \frac{1}{C}\Big(  \frac{1}{(-\Delta)^\alpha} \big[ \Delta^2(e_4 \Omega^i_{\CH} \psi)^2 +  (e_3 \Omega^i_{\CH} \psi)^2 \big] + \big[(e_1\Omega^i_{\CH} \psi)^2 +(e_2 \Omega^i_{\CH} \psi)^2 \big]\Big) \, \vol \\
&+ \sum\limits_{i=1}^3\int\limits_{\{r' \leq r \leq \rblue\} \cap \{f^- \leq u_1\}} \underbrace{\Big(\pi(\Omega^i_{\CH})^{\mu \nu} \nabla_\mu \nabla_\nu \psi + \nabla_\mu \pi(\Omega^i_{\CH})^{\mu \nu} \partial_\nu \psi - \frac{1}{2} \partial^\mu \big(\mathrm{tr} \pi(\Omega^i_{\CH})\big) \partial_\mu \psi\Big)}_{=:Err} w_\alpha X \Omega^i_{\CH}\psi \, \vol \\
\leq &  \sum\limits_{i=1}^3 \int\limits_{\Sigma_{\rblue} \cap \{f^+ \geq v_1\}} \T[\Omega^i_{\CH}\psi](w_\alpha X, (dr)^\sharp) \, \volr {\;.}
\end{split}
\end{equation}

First we note that the right hand side of \eqref{SecOrdEnergyEstBlue} is finite: indeed, along $\Sigma_{\rblue}$ one can write each $\Omega^i_{\CH}$, $i \in \{1,2,3\}$, as a linear combination of the ${\Omega^j_{\Hp}}$ with uniformly bounded coefficients such that, moreover, the derivatives of the coefficients with respect to the principal frame are also uniformly bounded. Hence, \eqref{BoundSecOrdEnInt} together with \eqref{DecayUpToRBlue} show the claim.

We want to show that \eqref{SecOrdEnergyEstBlue} implies
\begin{equation}
\label{SecOrderGronwallBlue}
\begin{split}
\sum_{i=1}^3&\int\limits_{\Sigma_{r'} \cap \{f^- \leq u_1\}} \T[\Omega^i_{\CH} \psi](w_\alpha X, (dr)^\sharp) \, \volr \\ 
&+ \sum\limits_{i=1}^3\int\limits_{\{r' \leq r \leq \rblue\} \cap \{f^- \leq u_1\}} \frac{1}{C'}\Big(  \frac{1}{(-\Delta)^\alpha} \big[ \Delta^2(e_4 \Omega^i_{\CH} \psi)^2 +  (e_3 \Omega^i_{\CH} \psi)^2 \big] + \big[(e_1\Omega^i_{\CH} \psi)^2 +(e_2 \Omega^i_{\CH} \psi)^2 \big]\Big) \, \vol \\
\lesssim &  1 + \sum\limits_{i=1}^3 \int\limits_{\{r' \leq r \leq \rblue\} \cap \{f^- \leq u_1\}} \T[\Omega^i_{\CH}\psi](w_\alpha X, (dr)^\sharp) \, \vol \;.
\end{split}
\end{equation}
Recalling $\vol = -dr \wedge \volr$ and using Gronwall then implies
\begin{equation*}
\sum_{i=1}^3\int\limits_{\Sigma_{r'} \cap \{f^- \leq u_1\}} \T[\Omega^i_{\CH} \psi](w_\alpha X, (dr)^\sharp) \, \volr \leq C \qquad \textnormal{ for all } r' \in (r_-, \rblue) \;.
\end{equation*}
Integrating this bound in $r$ {and substituting it into} \eqref{SecOrderGronwallBlue} then proves Proposition \ref{PropILED2}.

We will now prove \eqref{SecOrderGronwallBlue}. First note that one has
\begin{equation}
\label{ControlEnBlue}
\T[\Omega^i_{\CH} \psi](w_\alpha X, (dr)^\sharp) \gtrsim (e_3 \Omega_{\CH}^i \psi)^2 + \Delta^2\big[ (e_4 \Omega_{\CH}^i \psi)^2 + (e_1 \Omega_{\CH}^i \psi)^2 + (e_2 \Omega_{\CH}^i \psi)^2\big] \gtrsim w_\alpha^2(X \Omega_{\CH}^i \psi)^2 \;.
\end{equation}
Moreover, the following components of the deformation tensors $\pi(\Omega_{\CH}^i)$, $i = 1,2{,3}$, are needed:
\begin{equation}
\label{DefTen}
\begin{aligned}
\pi(\Omega_{\CH}^1)_{11} &=  2\sin \varphi_{\CH} \frac{a^2 \Si \Co}{\rho^2}      \;,
&&\pi(\Omega_{\CH}^2)_{11} = -2\cos \varphi_{\CH} \cdot \frac{a^2 \Si \Co}{\rho^2}     \;, \\
\pi(\Omega_{\CH}^1)_{22} &=   {-}2\sin \varphi_{\CH} a^2 \Co \Si (\frac{1}{r_-^2 + a^2} {+} \frac{1}{\rho^2})     \;,
&&\pi(\Omega_{\CH}^2)_{22} = 2\cos\varphi_{\CH} \cdot a^2 \Si\Co\big(\frac{1}{r_-^2 + a^2} + \frac{1}{\rho^2}\big)     \;, \\
\pi(\Omega_{\CH}^1)_{13} &= {\f{2a}{\rho}}\cos\varphi_{\CH}\big[   (1-\frac{r^2 + a^2}{r_-^2 + a^2}) -\frac{\Co^2\Delta}{{\rho^2}}\big] \;,
&&\pi(\Omega_{\CH}^2)_{13} = \frac{2a}{\rho} \sin \varphi_{\CH}\Big(\big[ 1 - \frac{r^2 + a^2}{r_-^2 + a^2}\big] - \frac{\Co^2 \Delta}{\rho^2}\Big)  \;,    \\
\pi(\Omega_{\CH}^1)_{23} &= \frac{2a\Co}{\rho}\sin\varphi_{\CH}\big[\frac{\Delta}{\rho^2} - (1 - \frac{r^2 + a^2}{r_-^2 + a^2})\big]       \;,
&&\pi(\Omega_{\CH}^2)_{23} = \frac{2a\Co}{\rho} \cos \varphi_{\CH} \Big(\big[ 1 - \frac{r^2 + a^2}{r_-^2 + a^2}\big] - \frac{\Delta}{\rho^2}\Big)  \;,  \\
\pi(\Omega_{\CH}^1)_{33} &= 0      \;,
&&\pi(\Omega_{\CH}^2)_{33} = 0   \;,  \\
\pi(\Omega_{\CH}^1)_{34} &= -\frac{4a^2 \Si\Co}{\rho^2} \sin\varphi_{\CH}       \;,
&&\pi(\Omega_{\CH}^2)_{34} = \frac{4a^2 \Si \Co}{\rho^2} \cos \varphi_{\CH}   \;, \\
\pi(\Omega_{\CH}^1)_{44} &= 0      \;,
&&\pi(\Omega_{\CH}^2)_{44} = 0  \;, \\[5mm]
\pi(\Omega_{\CH}^3)_{11} &=  0 \;,&&\pi(\Omega_{\CH}^3)_{22} =0 {\;,}\\
\pi(\Omega_{\CH}^3)_{13} &= \frac{2a \Co \Si \Delta}{\rho^3}\;, &&\pi(\Omega_{\CH}^3)_{23} = 0  \;,\\
\pi(\Omega_{\CH}^3)_{33} &= 0\;, &&\pi(\Omega_{\CH}^3)_{34} = 0 \;,\\
\pi(\Omega_{\CH}^3)_{44} &= 0 \;.
\end{aligned}
\end{equation}
We now estimate the third term of \eqref{SecOrdEnergyEstBlue} by Cauchy{--}Schwarz: each of the individual terms of $Err$ {(after squaring)} is weighted with an $\varepsilon> 0$, while the term $w_\alpha^2 (X \Omega^i_{\CH} \psi)^2$ is weighted with $\varepsilon^{-1}$. By \eqref{ControlEnBlue}, the latter term is controlled by $\T[\Omega^i_{\CH} \psi](w_\alpha X, (dr)^\sharp)$. We now show that the first terms are either controlled by Proposition \ref{PropFirstDerILED} or, for $\varepsilon$ small, can be absorbed in the second term of \eqref{SecOrdEnergyEstBlue}.

Let us first consider the term $\partial^\mu(\mathrm{tr} \pi\big(\Omega_{\CH}^i)\big)\partial_\mu \psi$ in $Err$. In order to control this term (after Cauchy{--}Schwarz) by Proposition \ref{PropFirstDerILED}, we need to show that all $e_4 \psi$ terms come with a coefficient that is $\mathcal{O}(r-r_-)$ for $r \to r_-$; i.e., we need to show that
\begin{equation*}
e_3\big(\mathrm{tr}{\pi}(\Omega^i_{\CH})\big) = e_3\big[\mathrm{tr}{\pi}(\Omega^i_{\CH})_{11} + \mathrm{tr}{\pi}(\Omega^i_{\CH})_{22} -\mathrm{tr}{\pi}(\Omega^i_{\CH})_{34}\big]  = \mathcal{O}(r-r_-)  \textnormal{ for } r \to r_- \;.
\end{equation*}
This, however, follows easily from \eqref{DefTen}, since $e_3$ either acts on a function of $r$, which generates a $\Delta$, or $e_3$ acts on $\varphi_{\CH}$, which also goes linearly to zero for $r \to r_-$, see \eqref{DerivativesPhi}.

The term $\nabla_\mu \pi(\Omega^i_{\CH})^{\mu \nu} \partial_\nu \psi$ in $Err$ is dealt with similarly. For $\pi = \pi(\Omega_{\CH}^i)$ we compute
\begin{equation*}
\begin{split}
&\big(\nabla_\mu \pi^{\mu \nu}\big)(e_3)_\nu \\
=& \langle \nabla_{e_1} \pi, e_1 \otimes e_3\rangle  + \langle  \nabla_{e_2} \pi , e_2 \otimes e_3\rangle  -\frac{1}{2}\langle \nabla_{e_3} \pi, e_4 \otimes e_3\rangle  - \frac{1}{2}\langle \nabla_{e_4} \pi, e_3 \otimes e_3\rangle  \\
=& e_1(\pi_{13}) + e_2(\pi_{23}) - \frac{1}{2} e_3(\pi_{43}) - \frac{1}{2}e_4(\pi_{33}) \\
&-\langle \pi, \big(\nabla_{e_1} e_1 + \nabla_{e_2} e_2 -\frac{1}{2} \nabla_{e_3}e_4 - \frac{1}{2}\nabla_{e_4} e_3\big) \otimes e_3\rangle  \\
&-\langle \pi, e_1 \otimes \nabla_{e_1} e_3 + e_2 \otimes \nabla_{e_2} e_3 - \frac{1}{2} e_4 \otimes \nabla_{e_3} e_3 - \frac{1}{2} e_3 \otimes \nabla_{e_4} e_3 \rangle  \\
=&e_1(\pi_{13}) + e_2(\pi_{23}) - \frac{1}{2} e_3(\pi_{43}) - \frac{1}{2}e_4(\pi_{33}) \\
& -\langle \pi, (\frac{{4}a^2 \Si\Co}{\rho^3} - \frac{\Co}{\Si}\frac{r^2 + a^2}{\rho^3}\big) e_1 \otimes e_3 + \frac{r\Delta}{\rho^4} e_1 \otimes e_1 + \frac{r\Delta}{\rho^4}e_2 \otimes e_2 + \frac{{2}ar\Si}{\rho^3} e_3 \otimes e_2 - \frac{r}{\rho^2} e_3 \otimes e_3 + \frac{r\Delta}{\rho^4} e_4 \otimes e_3\rangle  \;.
\end{split}
\end{equation*}
It is easy to check each of the terms and verify that $\big(\nabla_\mu \pi(\Omega_{\CH}^i)^{\mu \nu}\big)(e_3)_{\nu} = \mathcal{O}(r-r_-)$ for $r \to r_-$.

In order to discuss the term $\pi(\Omega_{\CH}^i)^{\mu \nu} \nabla_\mu \nabla_\nu \psi$, we again, as in Section \ref{BoundednessSecOrdEnBlueSec}, split the domain of integration in two regions $A$ and $A'$ and decompose $\pi(\Omega_{\CH}^i)^{\mu \nu} \nabla_\mu \nabla_\nu \psi$ in region $A$ with respect to the frame $\{e_3, e_4, e_1, e_2\}$ and in region $A'$ with respect to a frame $\{e_3, e_4, e_1', e_2'\}$, where again $e_1'$ and $e_2'$ form a smooth frame field for $\Pi^\perp$ in $A'$. As before we limit our discussion to the region $A$. 
\begin{enumerate}
\item We first consider the terms arising from $\pi(\Omega^i_{\CH})^{kl}$, where $k,l \in \{1,2\}$. The arising second order terms can be absorbed in the second term of \eqref{SecOrdEnergyEstBlue} {after choosing $\varepsilon$ to be sufficiently small}. To see that the arising first order terms (e.g.\ $\nabla_{e_1} e_1 \psi$) do not contain non-degenerate $e_4 \psi$ terms, we note that \eqref{CovDer} shows
\begin{equation*}
\langle \nabla_{e_k} e_l, e_3\rangle  = - \langle  e_l, \nabla_{e_k} e_3\rangle  = \mathcal{O}(r-r_-) \textnormal{ for } r \to r_-\;.
\end{equation*}
\item The second order terms arising from $\pi(\Omega^i_{\CH})^{\mu \nu}$, where $\mu = 3$, $\nu = 1,2$ can be absorbed (and controlled), and the arising first order terms do not contain $e_4$-derivatives since $\langle \nabla_{e_3} e_\nu, e_3\rangle  =0$.
\item It follows from \eqref{DefTen} that $\pi(\Omega^i_{\CH})^{\mu \nu}$, where $\mu = 4$, $\nu = 1,2$, is $\mathcal{O}(r - r_-)$ for $r \to r_-$, and hence the second order terms can be absorbed (and controlled) and the first order terms can be controlled.
\item It follows from \eqref{DefTen} that $\pi(\Omega^i_{\CH})^{\mu \nu} = 0$ for $\mu = \nu = 3$ and $\mu = \nu = 4$.
\item It remains to deal  with the terms arising from $\pi(\Omega^i_{\CH})^{34}$. Since $\psi$ satisfies the wave equation we have
\begin{equation*}
(e_3 \otimes e_4)^{\mu \nu} \nabla_{\mu} \nabla_{\nu} \psi = \big(g|_{\Pi^\perp}\big)^{\mu \nu} \nabla_\mu \nabla_\nu \psi \;,
\end{equation*}
and we have already shown how to deal with the terms on the right hand side in the first point.
\end{enumerate}
This finishes the proof of \eqref{SecOrderGronwallBlue} and hence Proposition \ref{PropILED2} is proved.

\subsection{Improved integrated energy decay for the {$e_1$ and $e_2$} derivatives}\label{sec.IILED}

{In this subsection, we complete the necessary stability estimates by improving the weights for $(e_1\psi)^2$ and $(e_2\psi)^2$ in the estimate in Proposition \ref{PropFirstDerILED} (see Proposition \ref{PropImprovedILEDAngular} below). Key to this improved integrated energy decay estimate is the following lemma, which is a Hardy inequality:}
{
\begin{lemma}\label{lemma.Hardy}
There exists $\tilde{r}\in (r_-,\rblue]$ sufficiently close to $r_-$ such that the following estimate holds for all smooth functions $\phi$ on $\mathcal M$:
\begin{equation*}
\int\limits_{\{r \leq \tilde{r}\} \cap \{f^- \leq u_1\}} \frac{1}{(-\Delta)^\alpha} \phi^2 \, \vol \leq C_\alpha \Big( \int\limits_{\Sigma_{\tilde{r}} \cap \{f^- \leq u_1\}} \phi^2 \, \volr + \int\limits_{\{r \leq \tilde{r}\} \cap \{f^- \leq u_1\}} \frac{1}{(-\Delta)^\alpha} (e_3 \phi)^2 \, \vol \Big) \;.
\end{equation*}
\end{lemma}
}
\begin{proof}
Let $h$ and $\phi$ be smooth functions on $\mathcal M$. The product rule gives
\begin{equation}
\label{ProductRule}
\int\limits_{\substack{\{r \leq \tilde{r}\} \cap \{f^- \leq u_1\} \\ \cap \{f^+ \leq v\}}} \big( \mathcal{L}_{e_3} h + h \, \mathrm{div} (e_3)\big) \phi^2 \, \vol = \int\limits_{\substack{\{r \leq \tilde{r}\} \cap \{f^- \leq u_1\} \\ \cap \{f^+ \leq v\}}} \mathcal{L}_{e_3}\big( h \phi^2 \, \vol\big) - \int\limits_{\substack{\{r \leq \tilde{r}\} \cap \{f^- \leq u_1\} \\ \cap \{f^+ \leq v\}}} 2h \phi \,(\mathcal{L}_{e_3} \phi) \, \vol \;.
\end{equation}
Setting $h = \frac{1}{(-\Delta)^\alpha}$, we compute
\begin{equation*}
e_3 h = -\frac{\alpha}{\rho^2} \partial_r \Delta \cdot \frac{1}{(-\Delta)^\alpha}
\end{equation*}
and
\begin{equation*}
\mathrm{div}(e_3) = \frac{2r\Delta}{\rho^4} + \partial_r \Big(\frac{\Delta}{\rho^2}\Big) \;.
\end{equation*}
We thus obtain
\begin{equation*}
e_3 h + h \, \mathrm{div}(e_3) = \frac{1}{(-\Delta)^\alpha} \Big(\frac{1 - \alpha}{\rho^2} \partial_r \Delta\Big) \;.
\end{equation*}
Hence, \eqref{ProductRule} together with Stokes' lemma yields
\begin{equation*}
\begin{split}
\int\limits_{\substack{\{r \leq \tilde{r}\} \cap \{f^- \leq u_1\} \\ \cap \{f^+ \leq v\}}} \frac{1- \alpha}{\rho^2}(-\partial_r \Delta) \cdot \frac{1}{(-\Delta)^\alpha} \phi^2 \, \vol &\leq \int\limits_{\substack{\{r \leq \tilde{r}\} \cap \{f^- \leq u_1\} \\ \cap \{f^+ \leq v\}}} \frac{2}{(-\Delta)^\alpha}\phi (e_3 \phi) \, \vol \\
&\qquad+ \int\limits_{\substack{\Sigma_{\tilde{r}} \cap \{f^- \leq u_1\} \\ \cap \{f^+ \leq v\}}} \frac{1}{(-\Delta)^\alpha} \phi^2 \cdot \Big(-\frac{\Delta}{\rho^2}\Big)\, \volr \;,
\end{split}
\end{equation*}
where we have dropped the negative boundary terms on the right hand side. Note that $(-\partial_r \Delta) (r_-) > 0$. Hence{, w}e can choose $\tilde{r}$ {sufficiently close} to $r_-$ so that after letting $v \to \infty$ and {using the} Cauchy{--}Schwarz {inequality} we obtain
\begin{equation*}
\int\limits_{\{r \leq \tilde{r}\} \cap \{f^- \leq u_1\}} \frac{1}{(-\Delta)^\alpha} \phi^2 \, \vol \leq C_\alpha \Big( \int\limits_{\Sigma_{\tilde{r}} \cap \{f^- \leq u_1\}} \phi^2 \, \volr + \int\limits_{\{r \leq \tilde{r}\} \cap \{f^- \leq u_1\}} \frac{1}{(-\Delta)^\alpha} (e_3 \phi)^2 \, \vol \Big) \;.
\end{equation*}
\end{proof}
{Using the above lemma, we obtain the following improved integrated energy decay:}
\begin{proposition}
\label{PropImprovedILEDAngular}
Under Assumptions \ref{Assumption1} and \ref{Assumption2} the following holds: Let $\alpha \in {[}0,1)$ and $u_1 \in \R$. Then there exists a constant $C> 0$ such that
\begin{equation}
\label{ImprovedILED}
\int\limits_{\{ r\leq \rblue\} \cap \{f^- \leq u_1\}} \frac{1}{(-\Delta)^\alpha} \Big( \Delta^2 \psi_4^2 + \psi_3^2 +  \psi_1^2 + \psi_2^2 \Big)\, \vol \leq C \;
\end{equation}
{for $\rblue \in (r_-, r_+)$ as in Proposition \ref{PropFirstDerILEDDecay}.}
\end{proposition}

\begin{proof}
{By Proposition \ref{PropFirstDerILED}, we only need to prove the estimate for $\psi_1^2$ and $\psi_2^2$. Moreover, using Proposition \ref{PropFirstDerILED} again, it suffices to prove the desired bound after restricting the domain of integration to $\{r\leq  \tilde{r}\}$. }

{Applying Lemma \ref{lemma.Hardy} with $\phi = \Omega_{\CH}^i \psi$ and summing in $i$, we obtain }
{
\begin{equation}\label{ImprovedILED.1}
\begin{split}
&\sum_{i=1}^3\int\limits_{\{r \leq \tilde{r}\} \cap \{f^- \leq u_1\}} \frac{1}{(-\Delta)^\alpha} (\Omega_{\CH}^i \psi)^2 \, \vol \\
\leq &C_\alpha \sum_{i=1}^3\Big( \underbrace{\int\limits_{\Sigma_{\tilde{r}} \cap \{f^- \leq u_1\}} (\Omega_{\CH}^i \psi)^2 \, \volr }_{=:I}+ \underbrace{\int\limits_{\{r \leq \tilde{r}\} \cap \{f^- \leq u_1\}} \frac{1}{(-\Delta)^\alpha} (e_3 \Omega_{\CH}^i \psi)^2 \, \vol }_{=:II}\Big) \;.
\end{split}
\end{equation}
By the definition of $\Omega^i_{\CH}$, we have 
\begin{equation}\label{Om.e12,compare}
(e_1\psi)^2+(e_2\psi)^2\sim \sum_{i=1}^3  (\Omega^i_{\CH}\psi)^2
\end{equation}
Returning to \eqref{ImprovedILED.1}, $I$ is bounded: This can be proven using \eqref{Om.e12,compare} and also the argument\footnote{Indeed, this can be proven by straightforward modifications of the vector field $\tilde{w}X$.} that leads up to \eqref{DecayUpToRBlue}. $II$ is also bounded by Propositions \ref{PropILED2}. Therefore, the left hand side of \eqref{ImprovedILED.1} is bounded, which then implies the desired conclusion by \eqref{Om.e12,compare}.}
\end{proof}

\section{Instability estimates}\label{sec.instability}

\subsection{From the event horizon to the {hypersurface} $\gamma_{{\sigma}}$}
\label{sec.instability.1}

\begin{proposition}
\label{PropLowBoundGamma}
Let $\sigma > 0$ be given. Under the assumptions of Theorem \ref{MainThm}, and also if $ii)$ is replaced by $ii')$ of Remark \ref{RemarkToThm}, we have
\begin{equation*}
\int\limits_{\gamma_\sigma \cap \{f^+ \geq v_k\}} \T[\psi]\big(-\Delta e_4,(-df_{\gamma_\sigma})^\sharp\big) \, \volg \gtrsim v_k^{-(q+\delta)} 
\end{equation*}
for $k \in \N$ big enough, where $v_k \in \R$ is a sequence with $v_k \to \infty$ for $k \to \infty$.
\end{proposition}

Let us remark that if $ii)$ is not replaced by $ii')$, then one even has
\begin{equation*}
\int\limits_{\gamma_\sigma \cap \{f^+ \geq \V\}} \T[\psi]\big(-\Delta e_4,(-df_{\gamma_\sigma})^\sharp\big) \, \volg \gtrsim \V^{-(q+\delta)} \;
\end{equation*}
{for every $V\geq 1$.}
This stronger statement is, however, not needed in what follows.

\begin{proof}
We only give the proof in the case that assumption $ii)$ of Theorem \ref{MainThm} is replaced by assumption $ii')$ of Remark \ref{RemarkToThm}. The other case is analogous, but easier.

Using \eqref{gInverse} and \eqref{CoordVecInFrame} we compute
\begin{equation}
\label{DFPlus}
(-df^+)^\sharp = -\frac{a\Si}{\rho} \, e_2 - \frac{r^2 + 2Mr + a^2}{2 \Delta} \, e_3 - \frac{\Delta}{2\rho^2} \, e_4 
\end{equation}
and
\begin{equation}
\label{HawkingInFrame}
\begin{split}
T_{\CH} &= (r_-^2 + a^2) \, \partial_t + a \, \partial_\varphi\\
&= \frac{a\Si}{\rho}(r^2 - r_-^2) \, e_2 - \frac{1}{2}(r_-^2 + a^2 \Co^2) \, e_3 - \frac{\Delta}{2 \rho^2} (r_-^2 + a^2 \Co^2) \, e_4 \;.
\end{split}
\end{equation}
Abbreviating $\T[\psi]$ with $\T$, this implies
\begin{equation}
\label{SigmaPlusEnergyControlled}
\begin{split}
&\big|\T[\psi]\big(T_{\CH},(-df^+)^\sharp\big)\big| \\
\lesssim & -\Delta|\T(e_2, e_2)| +|\T(e_2, e_3)| - \Delta | \T(e_4, e_2)| +\frac{1}{-\Delta} | \T(e_3, e_3)| + |\T(e_3, e_4)| + \Delta^2 |\T(e_4,e_4)| \\
\lesssim & \frac{1}{-\Delta} \psi_3^2 + (\Delta \psi_4)^2 + \psi_1^2 + \psi_2^2 \;.
\end{split}
\end{equation}
We now consider the sequence $w_k \in \R$ from assumption $ii')$ of Remark \ref{RemarkToThm}. Recalling $\vol = df^+ \wedge \volsp$ and using the pigeonhole principle, Proposition \ref{PropFirstDerILEDDecay} implies that there exists a sequence $v_k \in [\frac{1}{2}w_k, w_k]$ and a constant $C> 0$ such that
\begin{equation}
\label{ExtraDecaySigmaPlus}
\begin{split}
\int\limits_{\Sigma^+_{v_k} \cap \{r \geq \rblue\}} &\Big( (\tilde{e}_4 \psi)^2 + (\tilde{e}_3 \psi)^2 + (e_2 \psi)^2 + (e_1 \psi)^2 \Big) \, \volsp    \\
&+ \int\limits_{\Sigma^+_{v_k} \cap \{r \leq \rblue\} \cap \{f_{\gamma_\sigma} \leq 1\}} \Big( \frac{1}{(-\Delta)^\alpha} \big[ \Delta^2 (e_4 \psi)^2 + (e_3 \psi)^2\big] + (e_2 \psi)^2 + (e_1 \psi)^2 \Big)\, \volsp \leq C \cdot v_k^{-(q + 1)}
\end{split}
\end{equation}
holds for all $k \in \N$ big enough.

At the heart of the proof is the energy estimate in the region $\{f^+ \geq v_k\} \cap \{f_{\gamma_\sigma} \leq 1\}$ with multiplier $T_{\CH}$:
\begin{equation}
\label{EnEstKilling}
\begin{split}
&\int\limits_{\gamma_\sigma \cap \{f^+ \geq v_k\}} \T[\psi]\big(T_{\CH}, (-df_{\gamma_\sigma})^\sharp\big) \, \volg \\
= &\int\limits_{\Sigma^+_{v_k} \cap \{f_{\gamma_\sigma} \leq 1\}} \T[\psi]\big(T_{\CH}, (-df^+)^\sharp\big) \, \volsp + \int\limits_{\Hp \cap \{f^+ \geq v_k\}} \T[\psi]\big(T_{\CH}, (dr)^\sharp\big)\, \volr {\;.}
\end{split}
\end{equation}
\begin{figure}[h]
  \centering
  \def\svgwidth{6cm}
    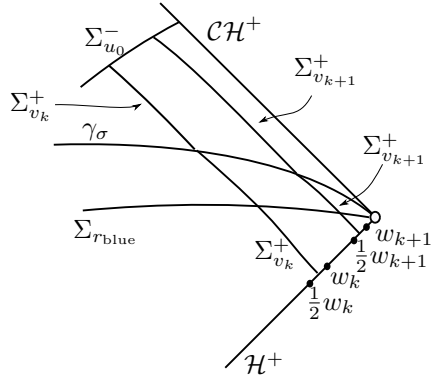
      \caption{Propagating the lower bound}
      \label{FigLastStep}
\end{figure}
We are going to show that the first term on the right hand side decays faster than the last term on the right hand side. It follows from \eqref{SigmaPlusEnergyControlled} that
\begin{equation}
\label{SupEst}
\begin{split}
&\int\limits_{\Sigma^+_{v_k} \cap \{f_{\gamma_\sigma} \leq 1\}} \big|\T[\psi]\big(T_{\CH}, (-df^+)^\sharp\big)\big| \, \volsp \\
\lesssim &\int\limits_{\Sigma^+_{v_k} \cap \{r \geq \rblue\}} \Big( (\tilde{e}_4 \psi)^2 + (\tilde{e}_3 \psi)^2 + (e_2 \psi)^2 + (e_1 \psi)^2 \Big) \, \volsp \\
&+ {\Big(}\sup\limits_{\substack{\Sigma^+_{v_k} \cap \{r \leq \rblue\} \\ \cap \{f_{\gamma_\sigma} \leq 1\}}} \frac{1}{(-\Delta)^{1-\alpha}}{\Big)} \cdot \int\limits_{\substack{\Sigma^+_{v_k} \cap \{r \leq \rblue\} \\ \cap \{f_{\gamma_\sigma} \leq 1\}}} \Big( \frac{1}{(-\Delta)^\alpha} \big[ \Delta^2 (e_4 \psi)^2 + (e_3 \psi)^2\big] + (e_2 \psi)^2 + (e_1 \psi)^2 \Big)\, \volsp \;.
\end{split}
\end{equation}
Moreover, \eqref{RCH} implies $-\Delta \gtrsim e^{\kappa_-(v_+ +v_-)}$ for $r \in [r_-, \rblue]$. Also using $f_{\gamma_\sigma} = v_+ + v_- - \sigma \log v_+ \leq 1$, we thus have
\begin{equation*}
\sup\limits_{\substack{\Sigma^+_{v_k} \cap \{r \leq \rblue\} \\ \cap \{f_{\gamma_\sigma} \leq 1\}}} \frac{1}{(-\Delta)^{1-\alpha}}  \lesssim \sup\limits_{\Sigma^+_{v_k}}\big(e^{-\kappa_- \sigma \log v_+}\big)^{1-\alpha} \lesssim {v_k}^{-\kappa_- \sigma (1 - \alpha)} \;.
\end{equation*}
Hence, \eqref{SupEst} together with \eqref{ExtraDecaySigmaPlus} yield
\begin{equation*}
\int\limits_{\Sigma^+_{v_k} \cap \{f_{\gamma_\sigma} \leq 1\}} \big|\T[\psi]\big(T_{\CH}, (-df^+)^\sharp\big)\big| \, \volsp  \leq C \cdot v_k^{-(q + 1) -\kappa_- \sigma(1-\alpha)}
\end{equation*}
for all $k \in \N$ big enough. Recall that assumption \eqref{AltLowBound} states $w_k^{-(q + \delta)} \lesssim \int_{\Hp \cap \{ v_+ \geq w_k \}} \T[\psi](T_{\CH}, \tilde{e}_4) \, \volr $. Since the integrand is manifestly non-negative for axisymmetric $\psi$, the same holds with ${\Hp \cap \{ v_+ \geq w_k \}}$ replaced by ${\Hp \cap \{ v_+ \geq v_k \}}$.
Putting these two things together, and using $ v_k \in [\frac{1}{2}w_k, w_k]$, we infer from \eqref{EnEstKilling} that
\begin{equation*}
\int\limits_{\gamma_\sigma \cap \{f^+ \geq v_k\}} \T[\psi]\big(T_{\CH}, (-df_{\gamma_\sigma})^\sharp\big) \, \volg  \geq C'\cdot  v_k^{-(q + \delta)} - C \cdot v_k^{-(q + 1) -\kappa_- \sigma(1-\alpha)}\;
\end{equation*}
{for some $C,C'> 0$.}
We can now choose $\alpha \in {[}0,1)$ sufficiently close to $1$, depending in particular on $\sigma > 0$ and $\delta \in [0,1)$, such that for all $k \in \N$ sufficiently large we have\footnote{Notice that the implicit constant here is allowed to depend on $\sigma$ and $\delta$.}
\begin{equation}
\label{LowerBoundKillingEn}
\int\limits_{\gamma_\sigma \cap \{f^+ \geq v_k\}} \T[\psi]\big(T_{\CH}, (-df_{\gamma_\sigma})^\sharp\big) \, \volg  \gtrsim  v_k^{-(q + \delta)} \;.
\end{equation}
We are going to write $T_{\CH}$ as a difference of a future and a past directed causal vector field.
Recalling \eqref{HawkingInFrame} and adding and subtracting $\frac{\Delta}{2 \rho^2} B \, e_4$, where $B$ is a positive constant, we obtain
\begin{equation*}
\begin{split}
T_{\CH} &= \frac{a\Si}{\rho}(r^2 - r_-^2) \, e_2 - \frac{1}{2}(r_-^2 + a^2 \Co^2) \, e_3 + \frac{\Delta}{2 \rho^2} B \, e_4 - \frac{\Delta}{2 \rho^2}(B + r_-^2 + a^2 \Co^2) \, e_4 \\
& =: Y - \frac{\Delta}{2 \rho^2}(B + r_-^2 + a^2 \Co^2) \, e_4 \;,
\end{split}
\end{equation*}
It now follows from
\begin{equation*}
\langle Y,Y\rangle  = \frac{a^2 \Si^2}{\rho^2} (r^2 - r_-^2)^2 + 2(r_-^2 + a^2 \Co^2) \frac{\Delta}{2\rho^2} B
\end{equation*}
that $0< B<\infty$ can be chosen big enough such that $Y$ is past directed timelike.
Since $\T[\psi]\big(Y, (-df_{\gamma_\sigma})^\sharp\big) \leq 0$ and $\frac{1}{2\rho^2}(B + r_-^2 + a^2 \Co^2)$ is bounded, \eqref{LowerBoundKillingEn} shows that
\begin{equation*}
\int\limits_{\gamma_\sigma \cap \{f^+ \geq v_k\}} \T[\psi]\big(-\Delta \, e_4, (-df_{\gamma_\sigma})^\sharp\big) \, \volg  \gtrsim  v_k^{-(q + \delta)} 
\end{equation*}
holds for all $k \in \N$ large enough. This finishes the proof of the proposition. 
\end{proof}

\subsection{From the {hypersurface} $\gamma_{{\sigma}}$ to the Cauchy horizon}
\label{sec.instability.2}

Starting from Proposition \ref{PropLowBoundGamma} we now finish the proof of Theorem \ref{MainThm}. At the centre of this step is the energy estimate with multiplier $L = -\frac{\Delta}{\rho^2} \, e_4$ in the region $\{f^- \leq u_0\} \cap \{f^+ \geq v_k\} \cap \{f_{\gamma_\sigma} \geq 1\}$ (see also Figure \ref{FigLastStep})\footnote{We briefly elaborate on how {one} derives this energy estimate by approximation by compact sets: using \eqref{DFPlus}, we obtain
\begin{equation*}
\big|\T[\psi]\big(-\Delta \, e_4, (-df^+)^\sharp\big) \big| \lesssim \Delta^2 \psi_4^2 + \psi_1^2 + \psi_2^2 \;.
\end{equation*}
It then follows from Proposition \ref{PropFirstDerILED} that there exists a sequence $w'_\ell \in \R$ with $w'_\ell \to \infty$ for $\ell \to \infty$ such that $\int_{\Sigma^+_{w'_\ell} \cap \{f^- \leq u_0\} \cap \{f_{\gamma_\sigma} \geq 1\}} \T[\psi]\big(-\Delta \, e_4, (-df^+)^\sharp\big)  \, \volsp \to 0$ for $\ell \to \infty$. Hence, one can consider the energy estimate in the region $\{f^- \leq u_0\} \cap \{w'_\ell \geq f^+ \geq v_k\} \cap \{f_{\gamma_\sigma} \geq 1\}$ and let $\ell$ tend to infinity. Moreover, a positive boundary term is dropped on the right hand side.}{, where $v_k \in \R$ is the sequence as given by Proposition \ref{PropLowBoundGamma}}:
\begin{equation}
\label{EnergyEste4}
\begin{split}
&\int\limits_{\Sigma^-_{u_0} \cap \{f^+ \geq v_k\}} \T[\psi]\big(L, (-df^-)^\sharp\big) \, \volsn \\
\geq &\int\limits_{\gamma_\sigma \cap \{f^+ \geq v_k\}} \T[\psi]\big(L, (-df_{\gamma_\sigma})^\sharp\big) \, \volg - \int\limits_{\substack{\{f^- \leq u_0\} \cap \{f^+ \geq v_k\} \\ \cap \{f_{\gamma_\sigma} \geq 1\}}} \T[\psi]_{\mu \nu} \pi(L)^{\mu \nu} \, \vol \;.
\end{split}
\end{equation}
We will show that one can choose $\sigma$ (and hence the auxiliary hypersurface $\gamma_\sigma$) such that the bulk term {(i.e., the second term)} on the right hand side decays faster than the first term on the right hand side.

It easily follows from \eqref{DefL} that
\begin{equation*}
\big|\T[\psi]_{\mu \nu} \pi(L)^{\mu \nu} \big| \lesssim \psi_1^2 + \psi_2^2 + \psi_3^2 + (\Delta \psi_4)^2 \;.
\end{equation*}
Moreover note that \eqref{RCH} implies $(-\Delta)^\alpha \lesssim e^{\alpha \kappa_-(v_+ + v_-)}$ for regions where $r$ is bounded away from $r_+$, and hence
\begin{equation*}
\sup\limits_{\{f_{\gamma_\sigma} \geq 1\} \cap \{f^+ \geq v_k\}} (- \Delta)^\alpha \lesssim \sup\limits_{\{f_{\gamma_\sigma} \geq 1\} \cap \{f^+ \geq v_k\}} e^{\alpha \kappa_- \sigma \log v_+} = \sup\limits_{\{f_{\gamma_\sigma} \geq 1\} \cap \{f^+ \geq v_k\}} v_+^{\sigma \kappa_- \alpha} \lesssim {v_k}^{\sigma \kappa_- \alpha} \;.
\end{equation*}
{In order to obtain \eqref{PolyLowBound}, it suffices to take $k$ sufficiently large and assume without loss of generality that $\{f^- \leq u_0\} \cap \{f^+ \geq v_k\} \cap \{f_{\gamma_\sigma} \geq 1\}\subset\{r\leq r_{blue}\}$.} Proposition \ref{PropImprovedILEDAngular} thus implies {for any}\footnote{with the implicit constant in $\lesssim$ depending on $\alpha$.} $\alpha\in {[}0,1)${, the following holds:} 
\begin{equation}\label{EnergyEste4.1}
\begin{split}
&\int\limits_{\substack{\{f^- \leq u_0\} \cap \{f^+ \geq v_k\} \\ \cap \{f_{\gamma_\sigma} \geq 1\}}} \big|\T[\psi]_{\mu \nu} \pi(L)^{\mu \nu} \big|\, \vol \\
\lesssim &{\Big(}\sup\limits_{\{f_{\gamma_\sigma} \geq 1\} \cap \{f^+ \geq v_k\}} (- \Delta)^\alpha{\Big)} \cdot \int\limits_{{\{ r\leq \rblue\} \cap \{f^- \leq u_0\}}} \frac{1}{(-\Delta)^\alpha} \Big( \Delta^2 \psi_4^2 + \psi_3^2 +  \psi_1^2 + \psi_2^2 \Big)\, \vol \lesssim {v_k}^{\sigma \kappa_- \alpha} \;.
\end{split}
\end{equation}
{On the other hand, by Proposition \ref{PropLowBoundGamma}, and using the fact that $\f{1}{\rho^2}$ is bounded uniformly above and below on $\mathcal M$, we can bound the first term on the right hand side of \eqref{EnergyEste4} as follows, with some $C> 0$:}
\begin{equation}\label{EnergyEste4.2}
\begin{split}
{\int\limits_{\gamma_\sigma \cap \{f^+ \geq v_k\}} \T[\psi]\big(L, (-df_{\gamma_\sigma})^\sharp\big) \, \volg\geq C \cdot v_k^{-(q+\delta)}.}
\end{split}
\end{equation}
{Combining \eqref{EnergyEste4}, \eqref{EnergyEste4.1} and \eqref{EnergyEste4.2},} we obtain
\begin{equation}
\label{FinalEst}
\int\limits_{\Sigma^-_{u_0} \cap \{f^+ \geq v_k\}} \T[\psi]\big(L, (-df^-)^\sharp\big) \, \volsn \geq C \cdot v_k^{-(q+\delta)}   - C' \cdot v_k^{\sigma \kappa_- \alpha}\;,
\end{equation}
{for some}\footnote{Notice that the constants $C$ and $C'$ here are allowed to depend on $\sigma$ and $\alpha$, but importantly, they are independent of $k$.} {$C,C'> 0$.} Without loss of generality we can assume that $\alpha \in (\frac{1}{2}, 1)$. We finally fix $\sigma > 0$ such that $\frac{1}{2} |\kappa_-| \sigma >  q + \delta$.  Note that this choice is independent of $\alpha$. Hence, \eqref{FinalEst} shows \eqref{PolyLowBound} and thus finishes the proof of Theorem \ref{MainThm}.

\bibliographystyle{acm}
\bibliography{Bibly}

\end{document}